\documentclass{article}

\usepackage{microtype}
\usepackage{graphicx}
\usepackage{subcaption}
\usepackage{booktabs} 
\usepackage{enumitem}
\usepackage{hyperref}

\usepackage[preprint]{icml2026}

\usepackage{amsmath}
\usepackage{amssymb}
\usepackage{mathtools}
\usepackage{amsthm}

\usepackage[capitalize,noabbrev]{cleveref}

\theoremstyle{plain}
\newtheorem{theorem}{Theorem}[section]

\newtheorem{lemma}[theorem]{Lemma}

\theoremstyle{definition}
\newtheorem{definition}[theorem]{Definition}

\theoremstyle{remark}

\usepackage[textsize=tiny]{todonotes}
\usepackage{multirow}
\usepackage{pifont}

\usepackage{xcolor}
\usepackage[table]{xcolor}
\usepackage{pifont} 

\icmltitlerunning{Black-box Audio Watermark Removal via Diffusion Priors}

\begin{document}

\twocolumn[
  \icmltitle{Black-box Audio Watermark Removal via Diffusion Priors}



  \icmlsetsymbol{equal}{*}

  \begin{icmlauthorlist}
    \icmlauthor{Lingfeng Yao}{houston}
    \icmlauthor{Xincong Zhong}{waseda}
    \icmlauthor{Chenpei Huang}{houston}
    \icmlauthor{Xuandong Zhao}{ucb}
    \icmlauthor{Hanqing Guo}{hawaii}\\
    \icmlauthor{Aohan Li}{uec}
    \icmlauthor{Jiang Liu}{waseda}
    \icmlauthor{Tomoaki Ohtsuki}{keio}
    \icmlauthor{Miao Pan}{houston}
  \end{icmlauthorlist}

  \icmlaffiliation{houston}{University of Houston}
  \icmlaffiliation{waseda}{Waseda University}
  \icmlaffiliation{ucb}{UC Berkeley}
  \icmlaffiliation{hawaii}{University of Hawaii at Mānoa}
  \icmlaffiliation{keio}{Keio University}
  \icmlaffiliation{uec}{The University of Electro-Communications}


  \icmlkeywords{Machine Learning, ICML}

  \vskip 0.3in
]



\printAffiliationsAndNotice{}  

\begin{abstract}
With the rise of AI-generated audio, watermarking has become widely used for detecting misuse and protecting intellectual property. However, adversaries may try to remove these watermarks, making it critical to evaluate how well watermarking schemes withstand removal attacks. Existing attacks are often impractical: they either noticeably degrade perceptual quality or require access to the watermarking scheme. We propose DiffErase, a black-box watermark removal attack that assumes no knowledge of the target watermarking scheme while maintaining perceptual quality. DiffErase perturbs watermarked audio to an intermediate diffusion noise level and regenerates it using a pretrained denoising model, effectively suppressing watermark signals. Theoretical analysis and extensive experiments demonstrate that inaudible audio watermarks are highly vulnerable: across multiple audio domains, DiffErase consistently removes watermarks while preserving perceptual quality. These findings highlight the need for future audio watermarking designs to consider diffusion-based threats. Code and demos are available at  \href{https://differase.github.io/DiffErase/}{https://differase.github.io/DiffErase/}.
\end{abstract}

\section{Introduction}
Audio generative models can now produce highly realistic audio that is indistinguishable to ordinary listeners. This technology reduces the cost of audio creation and enables many beneficial applications, but it also raises security risks related to misinformation, fraud, and public distrust. To mitigate these risks, proactive defenses such as watermarking~\citep{roman2024proactive, singh2024silentcipher, liu2024groot, huang2025echomark, chen2023wavmark, liu2023dear, yang2026fingerprinting, yang2026liteguard, liu2025securing} have been proposed to support governance of AI-generated content. By embedding an imperceptible signal into generated audio, a watermark provides a provenance cue that can be verified later, enabling detection of AI-generated audio.

Nevertheless, motivated adversaries may try to remove or disable watermarks to evade such verification. This makes it necessary to evaluate the robustness of watermarking schemes against removal attacks, ensuring they can withstand real-world threats. Prior work has studied attacks against audio watermarking. Some works~\citep{wen2025sok,o2025deep} apply signal-level transformations (e.g., additive noise or pitch shifting) or physical-level operations (e.g., re-recording) to disable detection. However, effective removal under these operations often introduces audible distortion, limiting their practicality. Other works explore more targeted strategies. Adversarial attacks~\citep{liu2024audiomarkbench, li2025harmonicattack} optimize perturbations added to the watermarked waveform to disable detection, but they typically require query access to detector output. Overwriting attacks~\citep{yao2025yours,liu2023detecting} embed an additional watermark to interfere with or replace the original, but they often depend on knowledge of the target watermarking system, such as the embedding architecture. These limitations raise a natural question: \emph{Can an adversary remove audio watermarks in a black-box setting while preserving perceptual quality?}

\begin{figure*}[!t]
\centering
\includegraphics[width=0.72\textwidth]{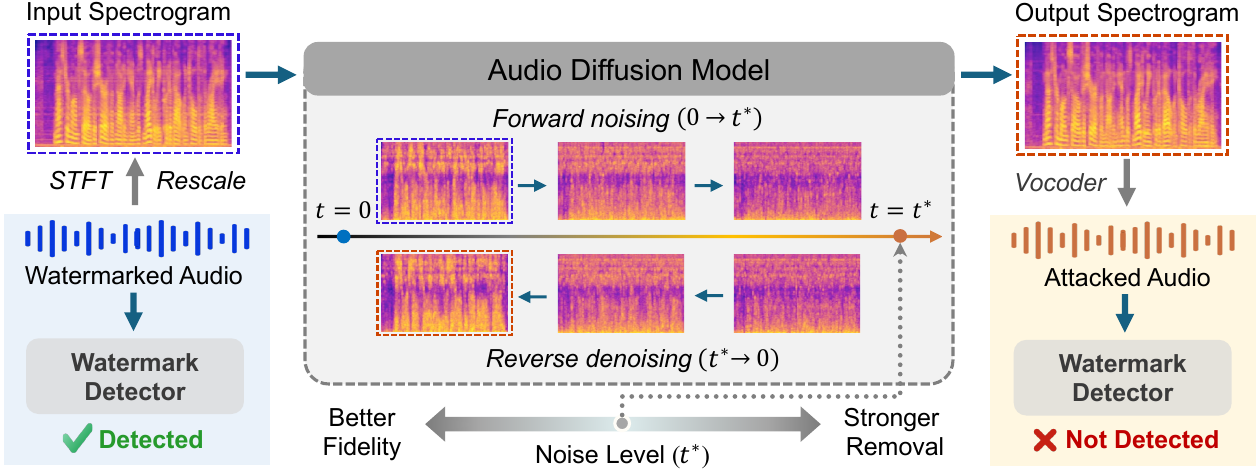}
\caption{\textbf{Overview of DiffErase.} Watermarked audio is converted to a mel-spectrogram via STFT, perturbed to an intermediate noise level $t^*$ via forward noising, and then denoised back to $t=0$. A vocoder reconstructs the attacked audio, which evades watermark detection. The noise level $t^*$ controls the trade-off between removal strength and perceptual fidelity.}
\label{fig:framework}
\end{figure*}

In this work, we investigate diffusion-based regeneration as a black-box watermark removal strategy. The key insight is that watermarking adds imperceptible perturbations that shift samples away from the natural data distribution, which detectors are trained to recognize. Diffusion models, trained to denoise corrupted inputs back to the data manifold, may naturally suppress these structured perturbations while preserving perceptual quality. Importantly, this approach requires no knowledge of the watermarking scheme, making it fully black-box. This intuition is supported by recent work on regeneration attacks against image watermarks~\citep{saberi2024robustness,zhao2024invisible} and diffusion-based defenses against adversarial examples~\citep{nie2022DiffPure, wu2023defending, guo2024wavepurifier}, where diffusion suppresses imperceptible perturbations. However, whether diffusion priors can enable black-box watermark removal for audio watermarking while preserving quality remains unexplored. 

Extending diffusion-based regeneration from images to audio requires careful design choices, as audio admits multiple representations with different reconstruction constraints. A direct application of waveform-level diffusion~\citep{kong2020diffwave} leads to over-smoothed content and temporal drift, degrading perceptual similarity. Alternatively, one can diffuse a linear spectrogram and invert it using the original phase~\citep{guo2024wavepurifier}, but the regenerated magnitude may be inconsistent with the preserved phase, producing audible artifacts. \emph{These issues make it challenging to obtain both strong watermark removal and high fidelity preservation}.

We propose DiffErase, a diffusion-based watermark removal attack operating in a strictly black-box setting. As illustrated in Figure~\ref{fig:framework}, DiffErase converts watermarked audio to a mel-spectrogram using STFT, perturbs it to an intermediate noise level $t^\ast$ via the forward noising process, and applies a pretrained denoising model to regenerate the mel-spectrogram back to $t=0$. A neural vocoder then reconstructs the waveform from the regenerated mel-spectrogram. The mel-spectrogram representation preserves salient perceptual structure such as energy contours and temporal envelopes, while enabling seamless integration with modern neural vocoders for high-quality reconstruction. We implement two complementary diffusion backbones: (i) mel-spectrogram diffusion operating directly in the mel-spectrogram domain, and (ii) latent diffusion works in a learned latent space from mel-spectrogram features. 

To understand why DiffErase succeeds, we provide a manifold-based analysis that models watermark embedding as an off-manifold perturbation. We show that diffusion reverse dynamics contract the watermark residue along the denoising trajectory, with an exponential decay bound controlled by the noise level. We evaluate DiffErase across three audio domains (speech, music, and environmental sounds) against five state-of-the-art watermarking systems (AudioSeal~\citep{roman2024proactive}, TimbreWM~\citep{liu2023detecting}, WavMark~\citep{chen2023wavmark}, Perth~\citep{resembleai_perth_2025}, and SilentCipher~\citep{singh2024silentcipher}). Our results demonstrate that DiffErase consistently disables watermark detection while maintaining high perceptual quality, highlighting diffusion-based regeneration as a practical threat that must be addressed in future watermark designs. 

Our contribution can be summarized as follows:
\begin{itemize}[leftmargin=*, itemsep=0pt, topsep=0pt]
    \item We propose DiffErase, a black-box diffusion-based attack that removes neural audio watermarks without requiring knowledge of the watermarking schemes.
    \item We provide a theoretical analysis explaining why diffusion dynamics suppress inaudible watermarks, with formal bounds on watermark contraction.
    \item We conduct extensive experiments demonstrating that DiffErase achieves strong removal performance across diverse audio domains and watermarking systems while preserving perceptual quality.
\end{itemize}

\section{Related Work}
\paragraph{Neural audio watermarking.}
Audio watermarking embeds a hidden signal into an audio carrier for downstream attribution, such as copyright verification and provenance tracking. A typical system consists of an \emph{embedder} that inserts a watermark (optionally carrying a message) and a \emph{detector} that verifies and/or recovers it. Practical methods must satisfy two key requirements: (i) \emph{fidelity}, meaning the watermark is imperceptible to listeners, and (ii) \emph{robustness}, meaning the watermark survives common manipulations. 

Traditional audio watermarking relies on handcrafted embedding and detection rules, making it difficult to balance fidelity and robustness. Recent neural approaches adopt end-to-end training, jointly optimized embedding, detection, fidelity, and robustness. Existing methods mainly differ in their embedding strategies. AudioSeal~\citep{roman2024proactive} and SilentCipher~\citep{singh2024silentcipher} employ a generator--detector design and embed watermarks into learned representations. WavMark~\citep{chen2023wavmark} and IDEAW~\citep{li2024ideaw} use invertible neural networks to model watermark embedding and detection as reversible transformations. TimbreWM~\citep{liu2023detecting} and DeAR~\citep{liu2023dear} embed watermark information in the frequency domain. While these methods demonstrate strong robustness against common distortions, their vulnerability to diffusion-based removal attacks remains largely unexplored.
\vspace{-2mm}

\paragraph{Attacks on audio watermarking.}
Existing attacks vary in their assumptions and effectiveness. Signal-level transformations~\citep{wen2025sok}, such as pitch shifting or lossy compression, require no knowledge of the watermark scheme but often degrade audio quality when applied strongly enough to remove watermarks. Overwriting attacks~\citep{yao2025yours} embed a second watermark to interfere with the original, but they typically assume knowledge of the target watermarking architecture. Adversarial attacks~\citep{liu2024audiomarkbench} formulate watermark removal as an optimization problem, iteratively crafting perturbations that fool the detector while preserving perceptual quality. However, they require query access to the detector, which may be impractical when the watermarking scheme is private. These limitations motivate our work: a removal attack that preserves perceptual quality under minimal assumptions.
\vspace{-2mm}

\paragraph{Diffusion models for audio.}
Diffusion models~\citep{ho2020denoising, song2020denoising} define a forward noising process and a learned reverse denoising process that maps noise back to the data distribution. DiffWave~\citep{kong2020diffwave} applies this framework to generate high-quality audio from mel-spectrograms and is commonly used as a neural vocoder. Conditional diffusion extends this paradigm to controlled generation, including text-to-audio~\citep{liu2023audioldm, yang2023diffsound} and image-to-audio synthesis~\citep{huang2023make}. Beyond generation, diffusion models have also been applied to audio restoration, reconstructing missing or corrupted content based on learned priors~\citep{wang2023audit}. Our work builds on diffusion-based regeneration, which has been explored as both a defense and an attack strategy. As a defense, DiffPure~\citep{nie2022DiffPure} and its audio-domain variants~\citep{wu2023defending, guo2024wavepurifier} leverage diffusion denoising to suppress adversarial perturbations. As an attack, \citet{zhao2024invisible} demonstrate that regeneration can remove invisible image watermarks by projecting perturbed samples back onto the natural data manifold. These findings suggest diffusion regeneration can suppress low-magnitude, imperceptible signals. However, prior work focuses exclusively on the image domain; whether this vulnerability extends to audio watermarking remains unexplored.

\vspace{-2mm}

\section{Problem Formulation and Threat Model}
\subsection{Problem formulation}
Let $x \in \mathbb{R}^T$ denote a clean audio waveform. An audio watermarking system consists of an embedding algorithm and a detection (or extraction) algorithm.

\paragraph{Embedding.} 
Given an audio signal $x$, an optional message $m$, and a secret key $k \in \mathcal{K}$ (equivalently, secret model parameters), the embedder outputs a watermarked waveform:
\begin{equation}
    x_w = \mathsf{Embed}(x, m, k) \in \mathbb{R}^T.
\end{equation}
For message-based schemes, $m \in \{0,1\}^L$ is an $L$-bit payload. Presence-only schemes embed a zero-bit watermark for detection, without carrying an explicit message.

\paragraph{Detection and extraction.}
For message-based schemes, an extractor recovers the embedded message:
\begin{equation}
    \widehat{m} = \mathsf{Ext}(x, k) \in \{0,1\}^L.
\end{equation}
If $x$ is watermarked with payload $m$, then $\widehat{m}$ should match $m$; otherwise $\widehat{m}$ is a random bit string. 

For presence-only schemes, a detector outputs a detection score indicating watermark presence:
\begin{equation}
    s = \mathsf{Det}(x,k) \in [0,1].
\end{equation}
A binary decision is obtained by thresholding:
\begin{equation}
    \mathsf{Result} = \mathbb{I}[s \ge \gamma] \in \{0,1\},
\end{equation}
where $\mathsf{Result}=1$ indicates watermark presence, $\mathsf{Result}=0$ indicates absence, and $\gamma$ is a detection threshold.

In this study, we evaluate our attack on five state-of-the-art open-source neural audio watermarking systems: AudioSeal~\citep{roman2024proactive}, TimbreWM~\citep{liu2023detecting}, WavMark~\citep{chen2023wavmark}, Perth~\citep{resembleai_perth_2025}, and SilentCipher~\citep{singh2024silentcipher}. 

\subsection{Threat model}
\paragraph{Adversary's capabilities.}
We consider a \emph{black-box} adversary who:
(i) has access only to the watermarked audio $x_w$;
(ii) has no knowledge of the watermarking scheme (architecture, weights, or hyperparameters), and cannot query the detector or extractor;
(iii) has sufficient computational resources to process $x_w$ using a diffusion model.

\paragraph{Adversary's objective.}
The adversary aims to produce a modified audio $\widehat{x}$ from $x_w$, disabling watermark verification while preserving perceptual quality. Formally, the attack succeeds if either
\begin{equation}
    \mathsf{Result}(\widehat{x},k)=0 \quad \text{or} \quad \mathrm{Acc}\big(\mathsf{Ext}(\widehat{x},k), m\big) < \eta,
\end{equation}
where $\mathrm{Acc}(\cdot,\cdot)$ denotes bit accuracy and $\eta$ is a threshold below which message recovery is considered failed. Meanwhile, perceptual quality should be maintained:
\begin{equation}
    \mathcal{Q}(\widehat{x}, x_w) \le q_0,
\end{equation}
where $\mathcal{Q}(\cdot,\cdot)$ is a perceptual distance metric and $q_0$ is a threshold beyond which degradation becomes perceptible.

\section{DiffErase Attack}
\subsection{Preliminaries: audio diffusion models}
Let $x_0 \in \mathbb{R}^d$ denote an audio representation. Diffusion models~\citep{ho2020denoising,song2020denoising} consist of (i) a fixed forward process that gradually adds noise and (ii) a learned reverse process that reconstructs data from noise.

\paragraph{Forward process.}
The forward process is a fixed Markov chain with variance schedule $\{\beta_t\}_{t=1}^N$. Let $\alpha_t\triangleq 1-\beta_t$ and $\bar{\alpha}_t\triangleq \prod_{s=1}^{t}\alpha_s$. The transition is
\begin{equation}
    q(x_t\mid x_{t-1})=\mathcal{N}\!\left(x_t;\sqrt{\alpha_t}\,x_{t-1},\,\beta_t \mathbf{I}\right),
\end{equation}
which yields a closed-form marginal distribution:
\begin{equation}
    q(x_t\mid x_0)=\mathcal{N}\!\left(x_t;\sqrt{\bar{\alpha}_t}\,x_0,\,(1-\bar{\alpha}_t)\mathbf{I}\right).
\end{equation}
Equivalently, $x_t$ can be sampled via reparameterization:
\begin{equation}
    x_t=\sqrt{\bar{\alpha}_t}\,x_0+\sqrt{1-\bar{\alpha}_t}\,\epsilon,
    \qquad \epsilon\sim\mathcal{N}(\mathbf{0},\mathbf{I}).
    \label{eq:reparam}
\end{equation}
As $t$ increases, $x_t$ approaches an isotropic Gaussian.

\paragraph{Reverse process.}
The reverse process is modeled as a learned Markov chain:
\begin{equation}
    p_\theta(x_{t-1}\mid x_t)=\mathcal{N}\!\left(x_{t-1};\mu_\theta(x_t,t),\,\Sigma_t\right),
\end{equation}
where $\Sigma_t$ is typically fixed. A standard parameterization predicts the forward noise using a neural network $\epsilon_\theta(x_t,t)$:
\begin{equation}
    \mu_\theta(x_t,t)=\frac{1}{\sqrt{\alpha_t}}
    \left(
        x_t-\frac{\beta_t}{\sqrt{1-\bar{\alpha}_t}}\,\epsilon_\theta(x_t,t)
    \right).
    \label{eq:reverse}
\end{equation}

The model is trained to minimize: 
\begin{equation}
    \mathcal{L}
    =
    \mathbb{E}_{x_0,\epsilon,t}
    \big[
        \|\epsilon-\epsilon_\theta(x_t,t)\|_2^2
    \big],
\end{equation}

\subsection{DiffErase: diffusion-based watermark removal}
We propose DiffErase, a diffusion-based attack that removes audio watermarks using a pre-trained diffusion model as a generative prior. Unlike standard diffusion-based generation, which starts from pure Gaussian noise, DiffErase follows an SDEdit-style procedure~\citep{meng2021sdedit} with two stages: (i) \emph{diffusion erasure}, which perturbs watermarked audio to an intermediate noise level, and (ii) \emph{semantic reconstruction}, which applies the learned reverse dynamics to recover a clean signal. Crucially, DiffErase requires no knowledge of the target watermarking systems.

A key challenge is selecting an appropriate audio representation for diffusion. Waveform-level diffusion produces over-smoothed content and temporal drift. Linear spectrogram diffusion generates magnitude inconsistent with the original phase, producing audible artifacts (see Appendix~\ref{apx:ablation_on_representation} for empirical comparisons). In contrast, mel-spectrograms capture salient structure such as energy contours and temporal envelopes while enabling high-quality reconstruction using modern vocoders, and we adopt this representation throughout. We model watermarked audio as $x_w = x_0 + \delta$, where $x_0$ is the clean audio and $\delta$ is the watermark perturbation.

\paragraph{Phase I: diffusion erasure (forward noising).}
We diffuse $x_w$ to an intermediate step $t^\ast\in\{1,\dots,N\}$ using the closed-form forward marginal in \eqref{eq:reparam}:
\begin{equation}
    x_{t^\ast}
    =
    \sqrt{\bar{\alpha}_{t^\ast}}\,x_w
    +
    \sqrt{1-\bar{\alpha}_{t^\ast}}\,\epsilon,
    \qquad
    \epsilon\sim\mathcal{N}(\mathbf{0},\mathbf{I}).
    \label{eq:diff_erase_forward}
\end{equation}
The hyperparameter $t^\ast$ controls a trade-off between watermark removal and reconstruction fidelity: a larger $t^\ast$ injects more noise, which better suppresses structured watermark signals but makes reverse reconstruction more challenging.

\paragraph{Phase II: semantic reconstruction (reverse denoising).}
Starting from $x_{t^\ast}$, we apply the reverse sampler induced by the pre-trained diffusion model from $t^\ast$ to $0$:
\begin{equation}
    x_{t-1}
    =
    \frac{1}{\sqrt{\alpha_t}}
    \left(
        x_t
        -
        \frac{\beta_t}{\sqrt{1-\bar{\alpha}_t}}
        \epsilon_\theta(x_t,t)
    \right)
    +
    \sigma_t z_t,
    \label{eq:diffusion_reverse}
\end{equation}
where $z_t\sim\mathcal{N}(\mathbf{0},\mathbf{I})$. We denote the final reconstruction by $\widehat{x}_0$. Since $\epsilon_\theta$ is trained on clean audio, the reverse dynamics tend to move samples toward high-density regions of the clean data distribution and suppress off-distribution perturbations introduced by watermarking. The complete procedure can be written as
\begin{equation}
\textsc{DiffErase}(x_w,t^\ast)
\triangleq
\textsc{Rev}(x_{t^\ast};\,t^\ast\!\rightarrow\!0),
\label{eq:differase_def}
\end{equation}
where $\textsc{Rev}(\cdot;\,t^\ast\!\rightarrow\!0)$ denotes reverse sampling as in \eqref{eq:diffusion_reverse}.

\paragraph{Instantiations.}\
We implement two variants of DiffErase, both operating on mel-spectrograms:

(i) \emph{Mel-spectrogram diffusion.} We apply DiffErase directly in the mel-spectrogram domain and reconstruct the waveform with a neural vocoder:
\begin{equation}
    \widehat{x}_0 = \textsc{Voc}\big(\textsc{DiffErase}(\textsc{Mel}(x_w), t^\ast)\big),
\end{equation}
where $\textsc{Mel}(\cdot)$ converts a waveform to a mel-spectrogram and $\textsc{Voc}(\cdot)$ inverts it back to a waveform.

(ii) \emph{Latent diffusion.} Following~\citet{rombach2022high}, we encode mel-spectrograms into a learned latent space using a pretrained variational autoencoder $(\textsc{Enc},\textsc{Dec})$, apply \textsc{DiffErase} in the latent space, and decode back:
\begin{equation}
\widehat{x}_0 = \textsc{Voc}\big(\textsc{Dec}(\textsc{DiffErase}(\textsc{Enc}(\textsc{Mel}(x_w)), t^\ast))\big).
\end{equation}
Latent diffusion reduces computational cost while achieving comparable attack performance. We compare these two variants in Section~\ref{exp:results_and_anlaysis}.

\subsection{Theoretical analysis}
\label{sec:theorectical_analysis}
We provide a theoretical justification for DiffErase. We model an imperceptible watermark as a small perturbation that moves the signal off the clean-audio manifold, and show that the diffusion prior suppresses these off-manifold components along the reverse trajectory.

\begin{definition}[$\Delta$-imperceptibility]
Let clean audio $x\in\mathbb{R}^{d}$ lie on a low-dimensional manifold $\mathcal{M}\subset\mathbb{R}^{d}$. Watermarked audio is constructed as $x_w = x + \delta$, where the perturbation satisfies $\|\delta\|_2\le \Delta$ for a small constant $\Delta > 0$ to ensure imperceptibility.
\end{definition}

\begin{definition}[$\tau$-margin detection]
Let $S(\cdot;k)\in\mathbb{R}$ denote a key-dependent detection statistic with key $k$. A watermark is \emph{detected} if $S(x_w;k)\ge\tau$ for threshold $\tau>0$. For processed audio $\widehat{x}$ derived from a watermarked input, we say the watermark is \emph{removed} if $S(\widehat{x};k)<\tau$.
\end{definition}

\paragraph{Off-manifold structure.}
Following~\citet{gilmer2018adversarial, stutz2019disentangling}, small perturbations that do not alter semantics tend to be orthogonal to the data manifold's tangent space $\mathcal{T}_x\mathcal{M}$.
While watermarks contain structure to enable decoding, they are designed to be imperceptible and statistically distinct from the natural audio distribution. 
We therefore decompose the watermark perturbation $\delta = \delta_{\parallel} + \delta_{\perp}$,  where $\delta_{\parallel} \in \mathcal{T}_x\mathcal{M}$ and $\delta_{\perp} \perp \mathcal{T}_x\mathcal{M}$, and assume the off-manifold component dominates, i.e., $\|\delta_{\perp}\|_2 \gg \|\delta_{\parallel}\|_2$.

\paragraph{Diffusion dynamics.}
We analyze the attack process by coupling the watermarked trajectory with a reference clean trajectory. The forward phase diffuses $x_w$ to an intermediate timestep $t^\ast$, which scales the signal by $\sqrt{\bar{\alpha}_{t^\ast}}$ and injects Gaussian noise. The reverse phase follows the deterministic dynamics given by the probability flow ODE with score model $s_\theta(x,t)\approx \nabla_x \log p_t(x)$, where $p_t$ is the marginal distribution at diffusion time $t$.

\begin{lemma}[Score restores off-manifold deviations]
\label{lem:direction}
Assume the manifold hypothesis and a local Gaussian approximation of $p_t$, the score function points towards the high-density region. For a watermarked state $x_t$ with off-manifold component $\Pi_{\perp}(x_t)\delta$, there exists $c_t>0$ such that
\begin{equation}
\langle s_\theta(x_t,t), \Pi_{\perp}(x_t)\delta\rangle \le -c_t \|\Pi_{\perp}(x_t)\delta\|_2^2,
\label{eq:off_manifold}
\end{equation}
\end{lemma}

\noindent
\emph{Proof sketch.} 
Since $p_t$ concentrates near $\mathcal{M}$, the marginal distribution in the normal direction approximates a Gaussian centered on the manifold. Consequently, the score function acts as a linear restoring force opposing off-manifold deviations. A detailed proof is provided in Appendix~\ref{app:proof_lemma_direction}.

\begin{lemma}[One-step contraction of watermark residue]
\label{lem:onestep}
Let $x_t$ and $x_t^{\mathrm{clean}}$ denote two coupled reverse trajectories initialized from the watermarked and clean states, respectively. Define the watermark residue at time $t$ as $r_t \triangleq x_t - x_t^{\mathrm{clean}}$.
Under Lemma~\ref{lem:direction} and the assumption that residue is dominated by its off-manifold component, there exists a contraction factor $\rho_t\in(0,1)$ such that
\begin{equation}
    \|r_{t-1}\|_2 \le \rho_t \|r_t\|_2 .
    \label{eq:onestep}
\end{equation}
\end{lemma}

\noindent
\emph{Proof sketch.}
Consider one reverse step for both trajectories and subtract them to obtain an update rule for $r_t$.
The score difference contributes a drift opposite to the normal component of $r_t$ (Lemma~\ref{lem:direction}). Since the off-manifold component dominates, this restoring effect yields the contraction in Eq.~\eqref{eq:onestep}. A detailed derivation is provided in Appendix~\ref{app:proof_lemma_onestep}.

\begin{theorem}[Exponential decay of watermark residue]
\label{thm:decay}
By combining the forward noising at timestep $t^\ast$ with the one-step contraction in Lemma~\ref{lem:onestep}, the final residue after reverse reconstruction satisfies
\begin{equation}
    \|r_0\|_2
    \le    \underbrace{\sqrt{\bar{\alpha}_{t^\ast}}}_{\text{forward scaling}}
    \cdot \underbrace{\left(\prod_{t=1}^{t^\ast}\rho_t\right)}_{\text{reverse contraction}}
    \cdot
    \Delta .
\end{equation}
Moreover, for any detection threshold $\tau>0$, there exists a minimum diffusion steps $t^\ast_{\min}$ such that for all $t^\ast > t^\ast_{\min}$, the watermark becomes undetectable ($S(\widehat{x_0};k)<\tau$).
\end{theorem}

Theorem~\ref{thm:decay} shows that DiffErase suppresses watermark residue via two complementary mechanisms: (i) forward noising scales the signal component by ($\sqrt{\bar{\alpha}_{t^\ast}}$), attenuating the initial watermark perturbation, and (ii) reverse reconstruction yields geometric contraction with factors ($\rho_t$), filtering out off-manifold components. The full proof and derivation of $t^\ast_{\min}$ are provided in Appendix~\ref{app:proof_thm_decay}.

\section{Evaluation}
\paragraph{Setup.}
We implement DiffErase with two diffusion backbones: (i) mel-spectrogram diffusion~\citep{von-platen-etal-2022-diffusers} with BigVGAN~\citep{lee2023bigvgan} as the vocoder, and (ii) latent diffusion on mel-spectrograms~\citep{liu2023audioldm} with HiFi-GAN~\citep{kong2020hifi} as the vocoder. Note that the diffusion process operates on integer timesteps $t\in\{1,\dots,N\}$ ($N=1000$ in our implementation), we report the noise level as a normalized ratio $t^\ast = t/N \in (0,1]$ for clarity. Training configurations and implementation details are provided in Appendix~\ref{apx:differase_detail}. 

\begin{table*}[!t]
\centering
\caption{\textbf{Comparison with baselines} on the speech domain. Left: audio quality metrics (higher is better). Right: watermark detection measured by $\mathrm{TPR}@1\%\mathrm{FPR}$ (lower is better); \ding{55} indicates $\mathrm{TPR} < 0.01$.}
\label{tab:main_results_speech_compact}
\setlength{\tabcolsep}{4pt}
\renewcommand{\arraystretch}{0.95}
\resizebox{0.95\textwidth}{!}{
\begin{tabular}{l l | c c c | c c c c c}
\toprule
\multirow{2}{*}{\textbf{Type}} & \multirow{2}{*}{\textbf{Attack}}
& \multicolumn{3}{c|}{\textbf{Audio Quality}}
& \multicolumn{5}{c}{\textbf{Watermark Detection}($\mathrm{TPR}@1\%\mathrm{FPR}\downarrow$)}\\
\cmidrule(lr){3-5}\cmidrule(lr){6-10}
& & SQUIM-MOS$\uparrow$ & ViSQOL$\uparrow$ & MUSHRA$\uparrow$
& AudioSeal & WavMark & TimbreWM & Perth & SilentCipher \\
\midrule

\multirow{5}{*}{\textbf{Signal-level}}
& Pitch shift        & 4.054 & 1.165 & 61.66 & \ding{55} & \ding{55} & \ding{55} & \ding{55} & \ding{55} \\
& Time stretch       & 4.072 & 1.502 & 66.25 & 1.00 & 0.95 & 1.00 & 1.00 & \ding{55} \\
& Low-pass filter    & 3.807 & 3.214 & 91.73 & 1.00 & 1.00 & 1.00 & 1.00 & 0.50 \\
& High-pass filter   & 2.757 & 1.579 & 73.20 & 1.00 & 1.00 & 1.00 & 1.00 & 0.53 \\
& Additive noise     & 3.062 & 1.063 & 25.64 & \ding{55} & \ding{55} & \ding{55} & \ding{55} & \ding{55} \\
\cmidrule(lr){1-10}

\multirow{2}{*}{\textbf{Codec}}
& MP3                & 4.503 & 4.123 & 96.42 & 1.00 & 0.97 & 1.00 & 1.00 & 0.34 \\
& EnCodec            & 4.369 & 3.708 & 96.97 & 1.00 & \ding{55} & \ding{55} & 0.50 & \ding{55} \\
\cmidrule(lr){1-10}

\multirow{3}{*}{\textbf{Adaptive}}
& Square Attack      & 3.025 & 2.567 & 54.07 & \ding{55} & \ding{55} & 0.28 & \ding{55} & \ding{55} \\
& \cellcolor{gray!15}\textbf{\textsc{DiffErase-latent}}
  & \cellcolor{gray!15}4.214 & \cellcolor{gray!15}3.477 & \cellcolor{gray!15}87.73
  & \cellcolor{gray!15}\ding{55} & \cellcolor{gray!15}\ding{55} & \cellcolor{gray!15}\ding{55}
  & \cellcolor{gray!15}\ding{55} & \cellcolor{gray!15}\ding{55} \\
& \cellcolor{gray!15}\textbf{\textsc{DiffErase-mel}}
  & \cellcolor{gray!15}4.423 & \cellcolor{gray!15}3.961 & \cellcolor{gray!15}93.81
  & \cellcolor{gray!15}\ding{55} & \cellcolor{gray!15}\ding{55} & \cellcolor{gray!15}\ding{55}
  & \cellcolor{gray!15}\ding{55} & \cellcolor{gray!15}\ding{55} \\
\bottomrule
\end{tabular}
}
\end{table*}

\paragraph{Datasets.}
We evaluate DiffErase across three audio domains: speech, music, and environmental sounds. For speech, we use the \textit{100-hour} subset of LibriSpeech~\citep{panayotov2015librispeech}. For music, we use a subset of \textit{FMA-small} from the Free Music Archive (FMA)~\citep{defferrard2016fma}. For environmental sounds, we use Clotho~\citep{drossos2020clotho}. We randomly sample 100 clips from each domain for evaluation (more details refer to Appendix~\ref{apx:dataset_detail}).

\paragraph{Target watermarking systems.}
We evaluate DiffErase against five state-of-the-art neural audio watermarking systems: AudioSeal~\citep{roman2024proactive}, TimbreWM~\citep{liu2023detecting}, WavMark~\citep{chen2023wavmark}, Perth~\citep{resembleai_perth_2025}, and SilentCipher~\citep{singh2024silentcipher}. All systems are configured according to their official releases. 

\paragraph{Attack baselines.}
We compare DiffErase against three categories of removal attacks. (i) \emph{signal-level attacks}, like pitch shifting, time stretching, filtering, and additive Gaussian noise; (ii) \emph{codec-based attacks}, including traditional codecs (e.g., MP3) and neural codecs (e.g., EnCodec); (iii) \emph{adaptive attacks}, specifically Square Attack~\citep{andriushchenko2020square}, a query-based adversarial attack. Details of all attack baselines are provided in Appendix~\ref{apx:attack_baseline}.

\paragraph{Evaluation metrics.}
We follow the evaluation protocol of~\citet{o2025deep}. For watermark removal, we report the true positive rate at a fixed false positive rate, denoted as $\mathrm{TPR}@1\%\mathrm{FPR}$; lower values indicate stronger removal. To assess perceptual quality after attack, we use both objective and subjective metrics: (1) SQUIM-MOS~\citep{kumar2023torchaudio}, a non-intrusive metric estimating mean opinion score (MOS) on a 1--5 scale without reference audio; (2) ViSQOL~\citep{chinen2020visqol}, which measures spectro-temporal similarity between reference and test audio on a 1--5 scale; and (3) MUSHRA, a subjective listening test where 16 participants rate samples on a 0--100 scale (details in Appendix~\ref{apx:subjective_test}). Since attackers have access only to the watermarked audio $x_w$, we use $x_w$ as the reference for computing perceptual metrics.

\subsection{Results and analysis}\label{exp:results_and_anlaysis}
\newcommand{\NADE}[2]{#1 / #2}

\begin{table}[t]
\centering
\caption{\textbf{\textsc{DiffErase-mel} performance across domains.} Each entry shows values before/after attack. MUSHRA (higher is better) measures subjective audio quality; $\mathrm{TPR}@1\%\mathrm{FPR}$ (lower is better) measures watermark detectability.}
\label{tab:differase_by_domain_metric_block}

\resizebox{0.98\columnwidth}{!}{%
\begin{tabular}{c c c c}
\toprule
\textbf{Domain} & \textbf{System} & MUSHRA$\uparrow$ & $\mathrm{TPR}@1\%\mathrm{FPR}\downarrow$ \\
\midrule

\multirow{5}{*}{Speech}
& AudioSeal      & \NADE{95.31}{93.19} & \NADE{1.00}{0.00} \\
& WavMark        & \NADE{98.38}{96.12} & \NADE{1.00}{0.00} \\
& TimbreWM         & \NADE{95.62}{95.06} & \NADE{1.00}{0.00} \\
& Perth          & \NADE{92.31}{90.69} & \NADE{1.00}{0.00} \\
& SilentCipher   & \NADE{96.69}{94.00} & \NADE{1.00}{0.00} \\

\midrule

\multirow{5}{*}{Music}
& AudioSeal      & \NADE{95.62}{87.12} & \NADE{1.00}{0.00} \\
& WavMark        & \NADE{92.75}{85.31} & \NADE{1.00}{0.00} \\
& TimbreWM         & \NADE{95.00}{84.06} & \NADE{1.00}{0.01} \\
& Perth          & \NADE{92.62}{84.94} & \NADE{1.00}{0.46} \\
& SilentCipher   & \NADE{93.31}{90.12} & \NADE{1.00}{0.00} \\

\midrule

\multirow{5}{*}{Env.}
& AudioSeal      & \NADE{92.00}{83.62} & \NADE{1.00}{0.00} \\
& WavMark        & \NADE{94.06}{87.94} & \NADE{1.00}{0.00} \\
& TimbreWM         & \NADE{93.62}{86.38} & \NADE{0.97}{0.00} \\
& Perth          & \NADE{90.69}{85.25} & \NADE{1.00}{0.19} \\
& SilentCipher   & \NADE{94.88}{89.19} & \NADE{1.00}{0.00} \\

\bottomrule
\end{tabular}%
}
\end{table}

\paragraph{Watermark removal and quality preservation.}
Table~\ref{tab:main_results_speech_compact} presents results on speech domain. Signal-level transformations generally fail to remove watermarks without causing severe quality degradation. Pitch shifting disables all watermark detectors but severely alters content and timbre, resulting in a low ViSQOL of 1.165 and MUSHRA of 61.66. Additive noise also removes all watermarks but introduces noticeable noise (MUSHRA:25.64). Other signal-level attacks, such as time stretching and frequency filtering, are largely ineffective, as these watermarking systems are typically trained to withstand such distortions.

Codec-based attacks demonstrate stronger performance. Most watermarking methods remain robust to MP3 compression, maintaining $\mathrm{TPR}@1\%\mathrm{FPR} \approx 1.00$. EnCodec achieves partial success against WavMark, TimbreWM, and SilentCipher, but fails against AudioSeal and Perth. Notably, both codecs preserve high perceptual quality (ViSQOL $>3.7$ and MUSHRA $>96$).

\begin{figure*}[!t]
    \centering
    \includegraphics[width=0.31\textwidth]{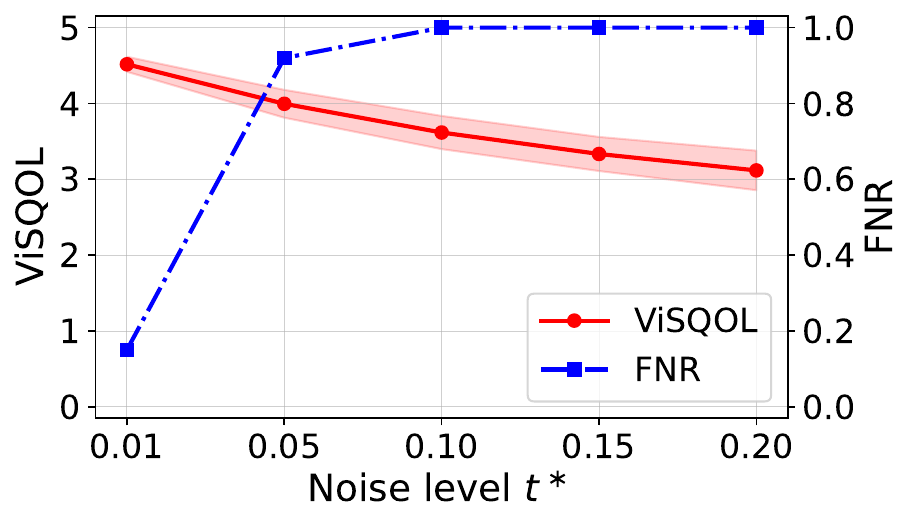}\hfill
    \includegraphics[width=0.31\textwidth]{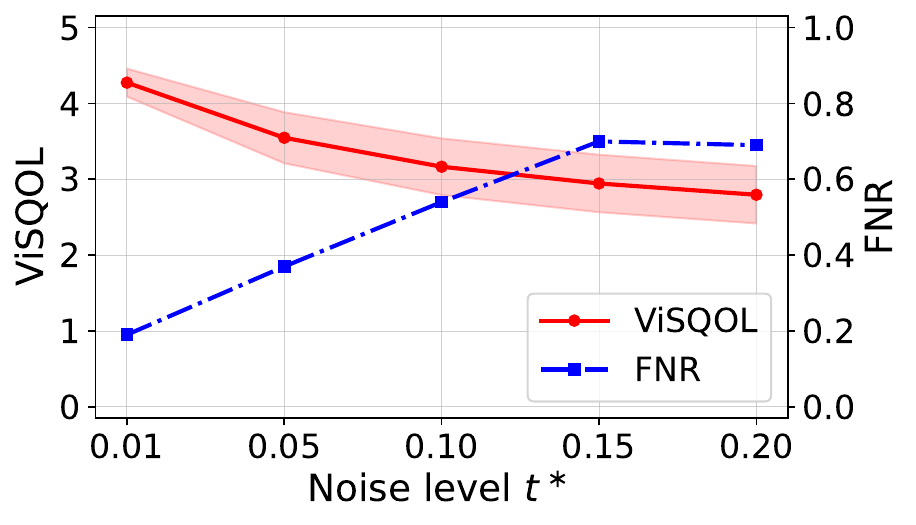}\hfill
    \includegraphics[width=0.31\textwidth]{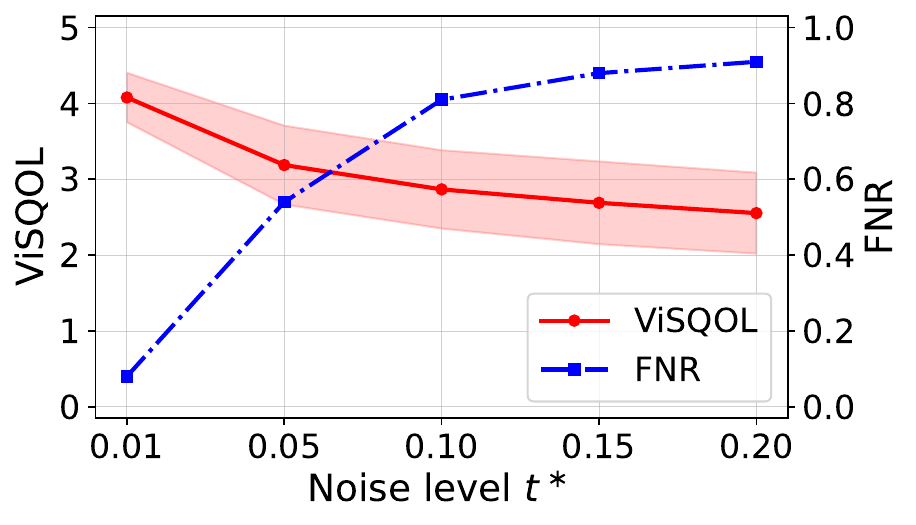}
    \caption{\textbf{Effect of noise level $t^*$ for \textsc{DiffErase-mel}.} Trade-off between audio quality (ViSQOL, left axis) and watermark removal (FNR $= 1 - \mathrm{TPR}$, right axis), evaluated on Perth. \textbf{Left:} Speech. \textbf{Middle:} Music. \textbf{Right:} Environment.}
    \label{fig:noise_level_tradeoff}
\end{figure*}

Adversarial attacks such as Square Attack degrade watermark detection by optimizing against detector outputs. However, they also introduce noticeable artifacts, yielding a MUSHRA score of only 54.07. In contrast, both DiffErase variants suppress watermark detection to $\mathrm{TPR}=0$ across all tested watermarking schemes while preserving perceptual quality. At noise level $t^\ast=0.1$, \textsc{DiffErase-latent} achieves SQUIM-MOS of 4.214 and ViSQOL of 3.477, and \textsc{DiffErase-mel} achieves SQUIM-MOS of 4.423 and ViSQOL of 3.961. Subjective MUSHRA scores further confirm that DiffErase maintains audio quality better than competing attacks. More comparison results on music and environmental sound are provided in Appendix~\ref{apx:comparison_on_music_env}.

Table~\ref{tab:differase_by_domain_metric_block} reports the overall performance of \textsc{DiffErase-mel} across all three domains. DiffErase consistently reduces $\mathrm{TPR}@1\%\mathrm{FPR}$ from 1.00 to near 0.00 for most watermarking systems, with limited perceptual degradation (MUSHRA drops of 2--10 points). The commercial watermark Perth exhibits stronger robustness: while DiffErase fully removes Perth watermarks on speech ($\mathrm{TPR}=0.00$), residual detection remains on music ($\mathrm{TPR}=0.46$) and environmental sounds ($\mathrm{TPR}=0.19$). We analyze watermark strength by measuring the $\ell_2$ distance between clean and watermarked audio. As shown in Figure~\ref{fig:l2_distance}, Perth induces substantially larger perturbations than all other watermarking schemes, approximately 4--10$\times$ higher, which explains its robustness to DiffErase and its impact on perceptual quality. 

\paragraph{Effect of noise level between watermark removal and audio quality.}
We evaluate how the noise level $t^\ast$ controls the trade-off between watermark removal and perceptual quality. As illustrated in Section~\ref{sec:theorectical_analysis} (Theorem~\ref{thm:decay}), a threshold $t^\ast_{\min}$ exists above which imperceptible watermarks can be removed. Larger $t^\ast$ provides stronger denoising suppression of the embedded signal, improving removal but degrading fidelity. Smaller $t^\ast$ preserves perceptual quality but cannot effectively reduce watermark. We use Perth as the target since it is the most robust watermark in our evaluation and induces the largest perturbation, as shown in Figure~\ref{fig:l2_distance}. Results are reported for \textsc{DiffErase-mel}; results of \textsc{DiffErase-latent} are provided in Appendix~\ref{apx:noise_level_tradeoff_latent}.

\begin{figure}[t]
  \centering
  \includegraphics[width=0.9\columnwidth]{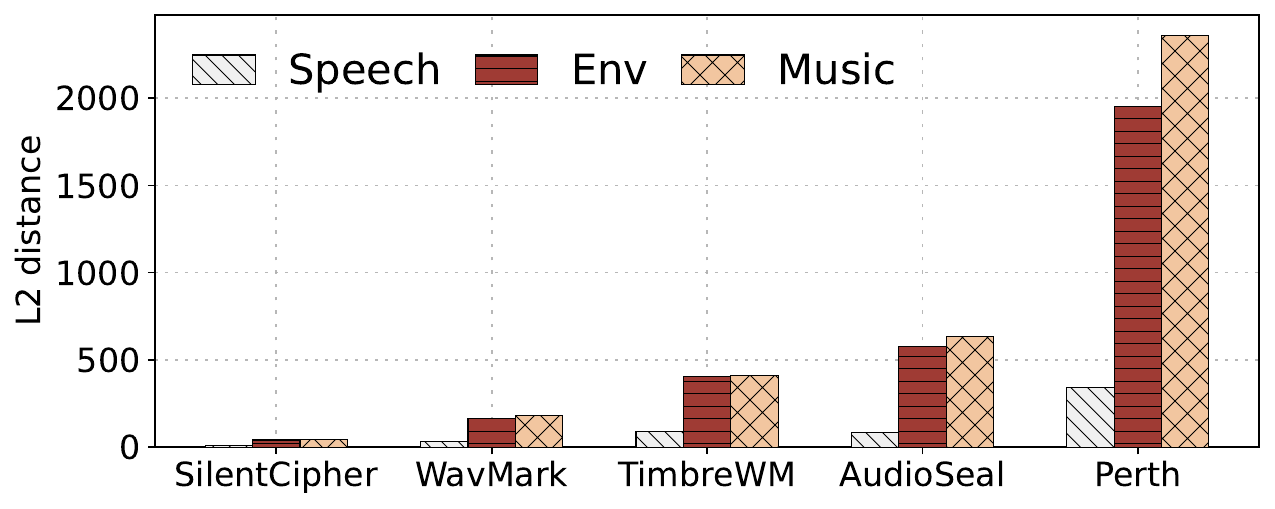}
  \caption{\textbf{The $\ell_2$ distance between clean and watermarked audio} across five watermarking methods on three domains. Perth embeds substantially stronger perturbations.}
  \label{fig:l2_distance}
\end{figure}
Figure~\ref{fig:noise_level_tradeoff} plots ViSQOL (left axis) and FNR $= 1-\mathrm{TPR}$ (right axis) as $t^\ast$ increases from 0.01 to 0.2. Across all domains, increasing $t^\ast$ improves watermark removal (higher FNR) while gradually reducing audio quality (lower ViSQOL). On speech, Perth becomes undetectable at $t^\ast \ge 0.10$ (FNR $\approx 1$) while quality remains high (ViSQOL $> 3.5$). On music and environmental sounds, FNR increases substantially with $t^\ast$ but cannot be fully removed even at $t^\ast = 0.20$. This is consistent with our earlier observation that Perth embeds a stronger watermark signal, which degrades perceptual quality but enhances robustness. It  requires a higher diffusion noise level to completely eliminate. 

\begin{figure*}[t]
\centering

\begin{subfigure}[t]{0.19\textwidth}
    \includegraphics[width=0.9\linewidth]{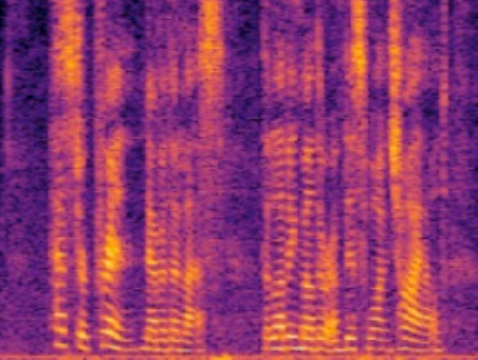}
    \caption*{Original}
\end{subfigure}
\hfill
\begin{subfigure}[t]{0.19\textwidth}
    \includegraphics[width=0.9\linewidth]{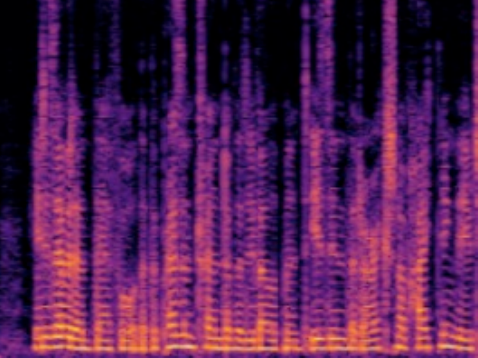}
    \caption*{Original}
\end{subfigure}
\hfill
\begin{subfigure}[t]{0.19\textwidth}
    \includegraphics[width=0.9\linewidth]{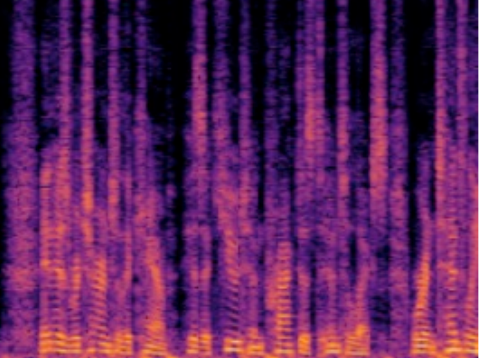}
    \caption*{Original}
\end{subfigure}
\hfill
\begin{subfigure}[t]{0.19\textwidth}
    \includegraphics[width=0.9\linewidth]{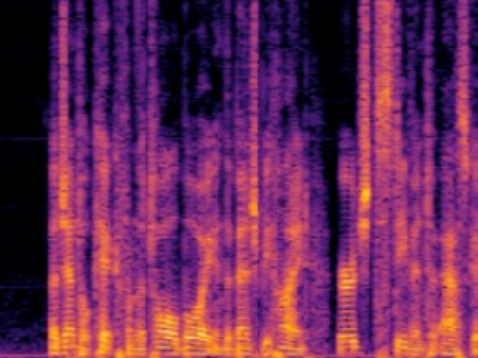}
    \caption*{Original}
\end{subfigure}
\hfill
\begin{subfigure}[t]{0.19\textwidth}
    \includegraphics[width=0.9\linewidth]{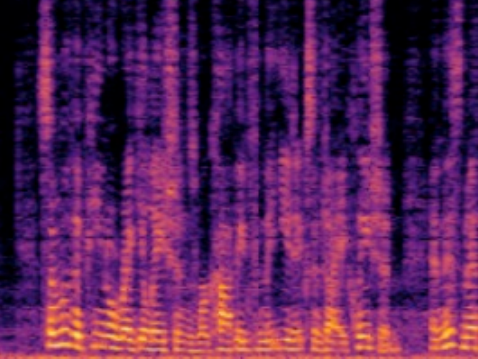}
    \caption*{Original}
\end{subfigure}

\vspace{2mm}

\begin{subfigure}[t]{0.19\textwidth}
    \includegraphics[width=0.9\linewidth]{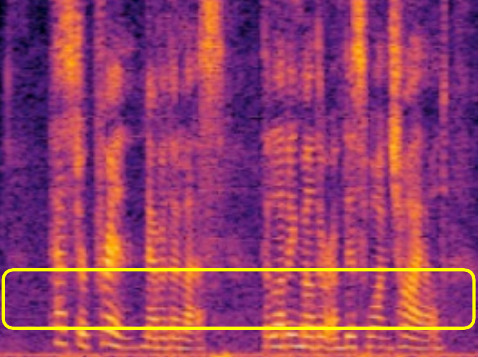}
    \caption*{AudioSeal}
\end{subfigure}
\hfill
\begin{subfigure}[t]{0.19\textwidth}
    \includegraphics[width=0.9\linewidth]{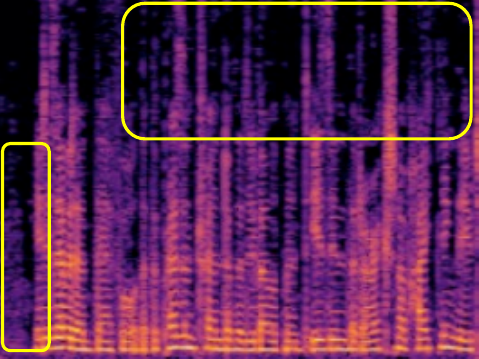}
    \caption*{WavMark}
\end{subfigure}
\hfill
\begin{subfigure}[t]{0.19\textwidth}
    \includegraphics[width=0.9\linewidth]{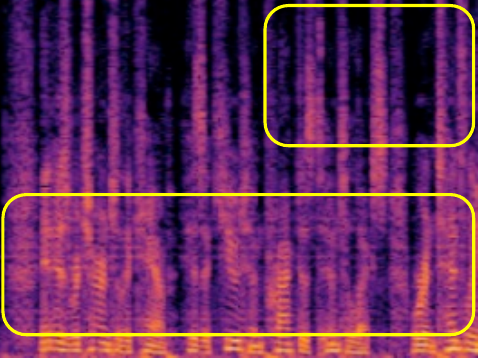}
    \caption*{TimbreWM}
\end{subfigure}
\hfill
\begin{subfigure}[t]{0.19\textwidth}
    \includegraphics[width=0.9\linewidth]{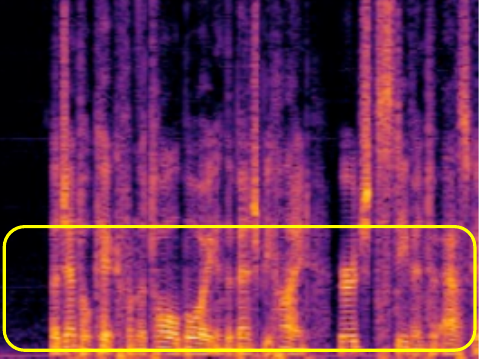}
    \caption*{Perth}
\end{subfigure}
\hfill
\begin{subfigure}[t]{0.19\textwidth}
    \includegraphics[width=0.9\linewidth]{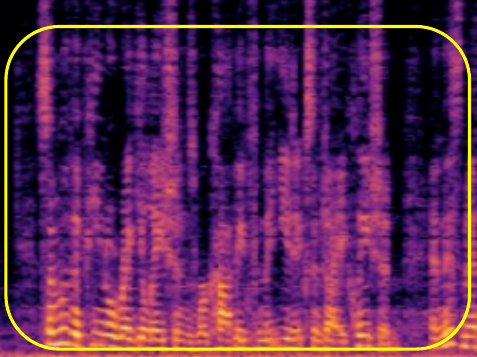}
    \caption*{SilentCipher}
\end{subfigure}

\vspace{2mm}

\begin{subfigure}[t]{0.19\textwidth}
    \includegraphics[width=0.9\linewidth]{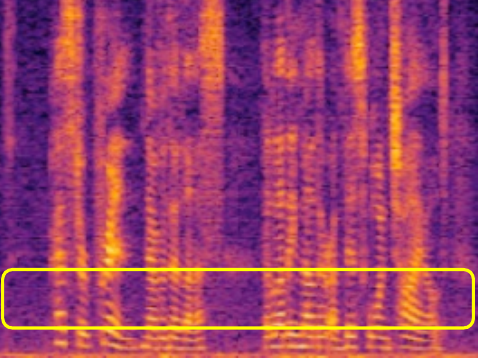}
    \caption*{DiffErase}
\end{subfigure}
\hfill
\begin{subfigure}[t]{0.19\textwidth}
    \includegraphics[width=0.9\linewidth]{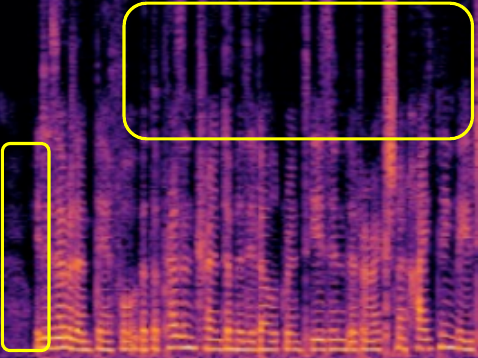}
    \caption*{DiffErase}
\end{subfigure}
\hfill
\begin{subfigure}[t]{0.19\textwidth}
    \includegraphics[width=0.9\linewidth]{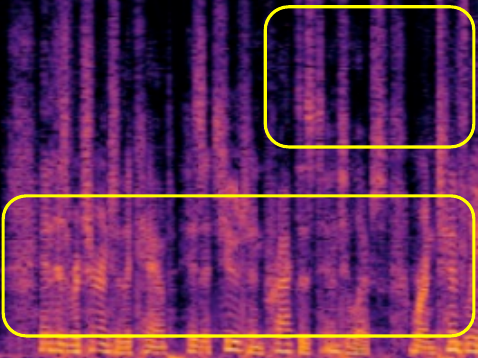}
    \caption*{DiffErase}
\end{subfigure}
\hfill
\begin{subfigure}[t]{0.19\textwidth}
    \includegraphics[width=0.9\linewidth]{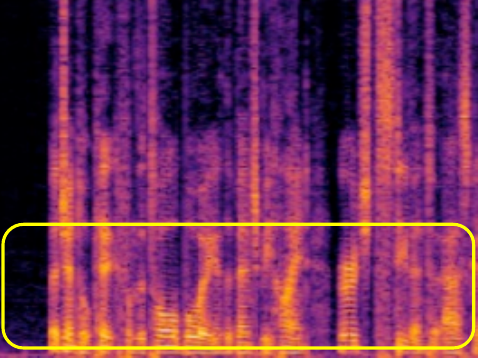}
    \caption*{DiffErase}
\end{subfigure}
\hfill
\begin{subfigure}[t]{0.19\textwidth}
    \includegraphics[width=0.9\linewidth]{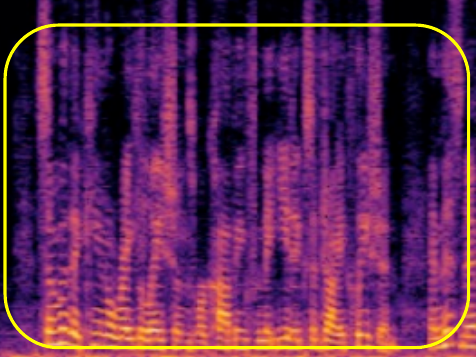}
    \caption*{DiffErase}
\end{subfigure}

\caption{\textbf{Spectrogram visualization.} Top: original audio. Middle: watermarked audio. Bottom: after \textsc{DiffErase-mel} ($t^* = 0.1$). Watermark patterns (middle row) are attenuated while acoustic content is preserved.}

\label{fig:spectrogram}
\end{figure*}

\paragraph{Spectrogram visualization.}
Figure~\ref{fig:spectrogram} visualizes spectrograms for five audio examples processed by \textsc{DiffErase-mel} at $t^\ast=0.1$. The top row shows original (clean) audio, the middle row shows watermarked samples, and the bottom row shows DiffErase outputs. Different watermarking schemes embed watermarks into different time--frequency regions (highlighted by yellow boxes). After DiffErase, these structured patterns are visibly attenuated or removed.

The watermark patterns vary across schemes. AudioSeal, TimbreWM, and Perth embed repeated patterns to ensure robustness, introducing visible band-like structures. These structures become much less pronounced after DiffErase. WavMark embeds small distortions in low-energy regions, making the watermark less perceptible but still detectable. DiffErase disrupts these localized patterns as well. SilentCipher shows only subtle and spatially spread perturbations with less obvious spectrogram patterns, which explains its lower robustness in our evalution. Across all samples, DiffErase eliminates watermark patterns while preserving the main acoustic structure, though some fine-grained details become slightly smoothed. This smoothing is controlled by $t^\ast$: smaller $t^\ast$ values better preserve details but may leave residual watermark evidence, while larger values improve removal at the cost of additional smoothing.

\begin{table}[h]
\centering
\caption{\textbf{Ablation study} on speech. Full results on music and environment are provided in Appendix~\ref{apx:more_ablation}.}

\label{tab:differase_ablation_speech}
\footnotesize
\setlength{\tabcolsep}{4pt}
\renewcommand{\arraystretch}{0.95}
\resizebox{\columnwidth}{!}{
\begin{tabular}{l c c c}
\toprule
\textbf{Method} & \textbf{TimbreWM} & \textbf{Perth} & \textbf{WavMark} \\
\midrule
\multicolumn{4}{l}{\textbf{(i) mel-to-waveform}} \\
\addlinespace[1pt]
Griffin-Lim algorithm (GLA) only                       & 1.00 & 1.00 & 1.00 \\
\textsc{DiffErase-mel} + GLA & 0.00 & 0.18 & 0.00  \\
\textsc{DiffErase-latent} + GLA & 0.00 & 0.00 & 0.00 \\
\cmidrule(lr){1-4}
\multicolumn{4}{l}{\textbf{(ii) diffusion sampler}} \\
\addlinespace[1pt]
\textsc{DiffErase-mel} (DDPM)  & 0.00 & 0.00 & 0.00 \\
\textsc{DiffErase-mel} (DDIM)  & 0.00 & 0.04 & 0.00 \\
\textsc{DiffErase-latent} (DDPM)  & 0.00 & 0.00 & 0.00 \\
\textsc{DiffErase-latent} (DDIM)  & 0.00 & 0.02 & 0.00 \\
\bottomrule
\end{tabular}
}
\end{table}

\paragraph{Ablation study.}
We conduct ablations to verify that watermark removal is primarily driven by the diffusion process rather than other pipeline components. We evaluate on three robust watermarking systems (TimbreWM, Perth, and WavMark), excluding AudioSeal and SilentCipher because their detection fails under waveform$\rightarrow$spectrogram$\rightarrow$waveform conversion alone.

To isolate the contribution of DiffErase, we reconstruct waveforms using the Griffin--Lim algorithm (GLA)~\citep{griffin1984signal}, which introduces minimal spectral distortion. As shown in Table~\ref{tab:differase_ablation_speech}, Griffin--Lim reconstruction alone does not affect watermark detection ($\mathrm{TPR}@1\%\mathrm{FPR}=1.00$ for all systems). In contrast, adding diffusion perturbations and denoising at $t^\ast=0.1$ substantially degrades detection: \textsc{DiffErase-mel} reduces TimbreWM and WavMark to $\mathrm{TPR}=0.00$ and Perth to 0.18, while \textsc{DiffErase-latent} removes all watermarks completely. The stronger removal of \textsc{DiffErase-latent} is consistent with the additional information bottleneck introduced from the VAE encoder. 

We also compare diffusion samplers. At the same noise level $t^\ast=0.1$, DDPM shows more effective removal than DDIM (50-step schedule), even though DDIM provides a faster processing speed. This suggests that the fine-grained denoising trajectory better suppresses watermark residues.

\section{Conclusion}
We propose DiffErase, a black-box attack that removes audio watermarks by leveraging diffusion models as generative priors. Unlike existing attacks, DiffErase requires neither detector queries nor knowledge of the watermarking schemes. Our theoretical analysis shows that diffusion dynamics suppress watermarks by contracting off-manifold perturbations along the reverse trajectory. Extensive experiments across three audio domains and five state-of-the-art watermarking systems demonstrate that DiffErase consistently removes watermarks while preserving perceptual quality. Our findings reveal a fundamental vulnerability in current audio watermarking designs: imperceptibility, while essential for practical deployment, inherently limits robustness against diffusion-based regeneration. This highlights the need for future watermarking designs to explicitly account for diffusion-based threats.

\clearpage

\section*{Impact Statement}
This paper studies the vulnerability of neural audio watermarking by presenting a black-box removal attack based on diffusion regeneration. The positive impact is to support more realistic robustness evaluation and facilitate stronger watermark designs. However, the proposed attack could potentially be misused to destroy provenance tracking or copyright protection. We are committed to responsible disclosure and encourage deployment-side mitigations and benchmarking protocols that account for diffusion-based threats.

\bibliography{icml26}
\bibliographystyle{icml2026}

\newpage
\clearpage
\appendix
\onecolumn
\section{Proofs in Section~\ref{sec:theorectical_analysis}}
\subsection{Proof of Lemma~\ref{lem:direction}}\label{app:proof_lemma_direction}
\begin{lemma}[Score restores off-manifold deviations]
Assume the manifold hypothesis and a local Gaussian approximation of $p_t$, the score function points towards the high-density region. For a watermarked state $x_t$ with off-manifold component $\Pi_{\perp}(x_t)\delta$, there exists $c_t>0$ such that
\begin{equation}
\langle s_\theta(x_t,t), \Pi_{\perp}(x_t)\delta\rangle \le -c_t \|\Pi_{\perp}(x_t)\delta\|_2^2,
\label{eq:off_manifold_apx}
\end{equation}
\end{lemma}

\begin{proof}
We rely on the geometric interpretation of the score function established in prior works~\citep{yoon2021adversarial, chung2022improving}. Under the manifold hypothesis, the diffusion marginal $p_t$ concentrates around $\mathcal{M}$ and can be locally approximated as a Gaussian centered at the manifold. As a result, for $x_t$ close to $\mathcal{M}$, the score has a normal component that points back toward the manifold. Concretely, we use the approximation
\begin{equation}
    \Pi_{\perp}(x_t) s_\theta(x_t, t) \approx -\frac{1}{\sigma_t^2} (x_t - \Pi_{\mathcal{M}}(x_t)),
\end{equation}
where $\Pi_{\mathcal{M}}(x_t)$ is the (local) projection of $x_t$ onto $\mathcal{M}$ and $\sigma_t$ is the noise scale at time $t$.

The watermark perturbation $\delta$ is modeled as predominantly off-manifold (as justified in Sec.~\ref{sec:theorectical_analysis}), the deviation from the manifold is dominated by the watermark component: $x_t - \Pi_{\mathcal{M}}(x_t) \approx \Pi_{\perp}(x_t)\delta$. Substituting this into the inner product, we obtain:
\begin{align}
    \langle s_\theta(x_t, t), \Pi_{\perp}(x_t)\delta\rangle 
    &= \langle \Pi_{\perp}(x_t) s_\theta(x_t, t), \Pi_{\perp}(x_t)\delta\rangle \\
    &\approx \left\langle -\frac{1}{\sigma_t^2}\Pi_{\perp}(x_t)\delta, \Pi_{\perp}(x_t)\delta \right\rangle \\
    &= -\frac{1}{\sigma_t^2} \|\Pi_{\perp}(x_t)\delta\|_2^2.
\end{align}

Letting $c_t = 1/\sigma_t^2>0$. Replacing the approximation by an inequality to account for local modeling error gives \eqref{eq:off_manifold_apx}.
\end{proof}

\subsection{Proof of Lemma~\ref{lem:onestep}}
\label{app:proof_lemma_onestep}
\begin{lemma}[One-step contraction of watermark residue]
Let $x_t$ and $x_t^{\mathrm{clean}}$ denote two coupled reverse trajectories initialized from the watermarked and clean states, respectively. Define the watermark residue at time $t$ as $r_t \triangleq x_t - x_t^{\mathrm{clean}}$.
Under Lemma~\ref{lem:direction} and the assumption that residue is dominated by its off-manifold component, there exists a contraction factor $\rho_t\in(0,1)$ such that
\begin{equation}
    \|r_{t-1}\|_2 \le \rho_t \|r_t\|_2 .
    \label{eq:onestep_apx}
\end{equation}
\end{lemma}

\begin{proof}
We analyze the deterministic probability flow ODE (as an idealized model for the reverse dynamics). Consider a discretized reverse step from time $t$ to $t-1$ using a first-order Euler solver:
\begin{equation}
    x_{t-1} \approx x_t + \eta_t s_\theta(x_t, t),
    \label{eq:euler_update}
\end{equation}
where $\eta_t > 0$ is an effective step size determined by the specific scheduler. Applying to both the watermarked trajectory $x_t$ and the clean trajectory $x_t^{\mathrm{clean}}$ using the same step size yields
\begin{align}
    r_{t-1} &= x_{t-1} - x_{t-1}^{\mathrm{clean}} \\
    &\approx (x_t + \eta_t s_\theta(x_t, t)) - (x_t^{\mathrm{clean}} + \eta_t s_\theta(x_t^{\mathrm{clean}}, t)) \\
    &= r_t + \eta_t \left( s_\theta(x_t, t) - s_\theta(x_t^{\mathrm{clean}}, t) \right).
    \label{eq:residue_update}
\end{align}

Next, we relate the score difference to the residue. Lemma~\ref{lem:direction} implies that, in normal directions near $\mathcal{M}$, the score acts as a restoring force. Under a first-order local linearization around the clean trajectory, we approximate
\begin{equation}
    s_\theta(x_t, t) - s_\theta(x_t^{\mathrm{clean}}, t) \approx -c_t (x_t - x_t^{\mathrm{clean}}) = -c_t r_t,
\end{equation}
where $c_t \approx 1/\sigma_t^2 > 0$ determines the strength of the restoring force derived in Lemma~\ref{lem:direction}.

Substituting this relationship into Eq.~\eqref{eq:residue_update}:
\begin{equation}
    r_{t-1} \approx r_t - \eta_t c_t r_t = (1 - \eta_t c_t) r_t.
\end{equation}
Taking the Euclidean norm, we obtain:
\begin{equation}
    \|r_{t-1}\|_2 \approx |1 - \eta_t c_t| \|r_t\|_2.
\end{equation}
Define the decay factor $\rho_t = |1 - \eta_t c_t|$. In standard diffusion schedules, the step size $\eta_t$ satisfies $0<\eta_t c_t<1$ (typically $0<\eta_t c_t<2$). Therefore, there exists $\rho_t \in (0, 1)$ such that:
\begin{equation}
    \|r_{t-1}\|_2 \le \rho_t \|r_t\|_2.
\end{equation}
\end{proof}

\subsection{Proof of Theorem~\ref{thm:decay}}
\label{app:proof_thm_decay}
\begin{theorem}[Exponential decay of watermark residue]
By combining the forward noising at timestep $t^\ast$ with the one-step contraction in Lemma~\ref{lem:onestep}, the final residue after reverse reconstruction satisfies
\begin{equation}
    \|r_0\|_2
    \le    \underbrace{\sqrt{\bar{\alpha}_{t^\ast}}}_{\text{Forward scaling}}
    \cdot \underbrace{\left(\prod_{t=1}^{t^\ast}\rho_t\right)}_{\text{Reverse contraction}}
    \cdot
    \Delta .
\end{equation}
Moreover, for any detection threshold $\tau>0$, there exists a minimum diffusion steps $t^\ast_{\min}$ such that for all $t^\ast > t^\ast_{\min}$, the watermark becomes undetectable ($S(\widehat{x}_0;k)<\tau$).
\end{theorem}

\begin{proof}
We first analyze the residue at the attack timestep $t^\ast$. Let $x_0$ denote the clean audio and $x_0^w = x_0 + \delta$ the watermarked audio, with $\|\delta\|_2 \le \Delta$. We adopt a coupled forward process where both trajectories share the same Gaussian noise $\epsilon \sim \mathcal{N}(\mathbf{0},\mathbf{I})$. At timestep $t^\ast$,

\begin{align}
    x_{t^\ast} &= \sqrt{\bar{\alpha}_{t^\ast}} x_0 + \sqrt{1-\bar{\alpha}_{t^\ast}} \epsilon, \\
    x_{t^\ast}^w &= \sqrt{\bar{\alpha}_{t^\ast}} (x_0 + \delta) + \sqrt{1-\bar{\alpha}_{t^\ast}} \epsilon.
\end{align}
Define the residue at $t^\ast$ as $r_{t^\ast} \triangleq x_{t^\ast}^w - x_{t^\ast}$. Then
\begin{equation}
    r_{t^\ast} = \sqrt{\bar{\alpha}_{t^\ast}} \delta.
\end{equation}
Consequently, the norm of the perturbation at $t^\ast$ is scaled by the signal decay factor:
\begin{equation}
    \|r_{t^\ast}\|_2
    =
    \sqrt{\bar{\alpha}_{t^\ast}}\,\|\delta\|_2
    \le
    \sqrt{\bar{\alpha}_{t^\ast}}\,\Delta .
    \label{eq:forward_bound}
\end{equation}

Applying Lemma~\ref{lem:onestep} iteratively from $t^\ast$ down to $0$ yields
\begin{equation}
\|r_0\|_2 \le \rho_1 \|r_1\|_2 \le \dots \le \left(\prod_{j=1}^{t^\ast}\rho_j\right)\|r_{t^\ast}\|_2.
\end{equation}
Substituting \eqref{eq:forward_bound} gives
\begin{equation}
\|r_0\|_2 \le \sqrt{\bar{\alpha}_{t^\ast}} \left(\prod_{j=1}^{t^\ast}\rho_j\right) \Delta.
\label{eq:product_decay}
\end{equation}
\end{proof}

\paragraph{Bound on diffusion steps $t^\ast_{\min}$.}
To derive the lower bound on $t^\ast$ required for watermark removal $S(\widehat{x}_0;k)<\tau$, we link the residue norm to the detection statistic. Assume the detection function $S(\cdot; k)$ is $L$-Lipschitz continuous. For clean audio (no watermark), $S(x^{\mathrm{clean}}; k) \approx 0$. For the reconstructed output $\widehat{x}_0$, we have
\begin{align}
    S(\widehat{x}_0;k)
    &\le
    |S(\widehat{x}_0;k)-S(x^{\mathrm{clean}};k)|
    + S(x^{\mathrm{clean}};k) \nonumber\\
    &\le
    L\|\widehat{x}_0-x^{\mathrm{clean}}\|_2 + S(x^{\mathrm{clean}};k) \nonumber\\
    &\approx
    L\|r_0\|_2 \nonumber\\
    &\le
    L\Delta \sqrt{\bar{\alpha}_{t^\ast}}
    \left(\prod_{t=1}^{t^\ast}\rho_t\right),
    \label{eq:S_bound}
\end{align}
where the final inequality uses \eqref{eq:product_decay}.

Let $\bar{\rho}\triangleq \max_{t\in\{1,\dots,t^\ast\}}\rho_t < 1$, so that $\prod_{t=1}^{t^\ast}\rho_t \le \bar{\rho}^{t^\ast}$. A sufficient condition for $S(\widehat{x}_0;k)<\tau$ is therefore
\begin{equation}
    L\Delta \sqrt{\bar{\alpha}_{t^\ast}}\, \bar{\rho}^{t^\ast} < \tau.
    \label{eq:tau_condition}
\end{equation}

Taking logarithms yields
\begin{equation}
    \ln(L\Delta) + \frac{1}{2}\ln \bar{\alpha}_{t^\ast} + t^\ast \ln \bar{\rho} < \ln \tau.
\end{equation}

Since $\bar{\alpha}_{t^\ast}\in(0,1)$ and $\bar{\rho}\in(0,1)$, their logarithms are negative. Rearranging gives the equivalent condition
\begin{equation}
    t^\ast |\ln \bar{\rho}| + \tfrac{1}{2}|\ln \bar{\alpha}_{t^\ast}|
    >
    \ln\!\left(\frac{L\Delta}{\tau}\right).
    \label{eq:t_min_condition}
\end{equation}

Because $\bar{\alpha}_{t^\ast}$ decreases monotonically with $t^\ast$, both terms on the left-hand side increase with $t^\ast$. Hence there exists a minimal $t^\ast_{\min}$ satisfying \eqref{eq:t_min_condition}, and any $t^\ast \ge t^\ast_{\min}$ is sufficient to ensure $S(\widehat{x}_0;k)<\tau$ under the stated assumptions.

\section{Implementation Details}
\subsection{DiffErase instantiations}\label{apx:differase_detail}
We implement DiffErase with two complementary backbones:

\paragraph{\textsc{DiffErase-mel}.}
This variant is built on \textsc{diffuser} library~\citep{von-platen-etal-2022-diffusers} and performs diffusion directly in the mel-spectrogram domain. We use a UNet2DModel as the denoiser, treating the mel-spectrogram as a single-channel image. The mel-spectrogram is computed with 80 mel bins. For waveform reconstruction, we use BigVGAN~\citep{lee2023bigvgan} (\texttt{bigvgan\_v2\_22khz\_80band\_256x}) as the vocoder.

\paragraph{\textsc{DiffErase-latent}.}
This variant is built on the AudioLDM pipeline~\citep{liu2023audioldm}. Mel-spectrograms (64 mel bins) are first encoded into a latent space using a pretrained AutoencoderKL with latent channel dimension of 8. Diffusion is then performed in the VAE latent space using a UNet with the following configuration: image size 64, base channels 128, 2 residual blocks per stage, channel multipliers $[1, 2, 3, 5]$, and attention at resolutions $\{8, 4, 2\}$. The reconstructed mel-spectrogram is converted back to waveform using HiFi-GAN~\citep{kong2020hifi}.

\subsection{Baseline attacks}\label{apx:attack_baseline}
\paragraph{Signal-level attacks.}
Following the setting in \citet{o2025deep}, pitch shifting uses a random shift in $[-1,1]$ semitones, and time stretching applies a random speed factor in $[0.95, 1.05]$. For frequency filtering, we apply a low-pass filter with a 4000\,Hz cutoff and a high-pass filter with a 500\,Hz cutoff. Additive noise is Gaussian with standard deviation $\sigma = 0.01$.

\paragraph{Codec-based attacks.}
We evaluate a traditional codec and a neural-based codec. MP3 compression is performed using FFmpeg at 32\,kbps. EnCodec~\citep{defossez2022high} uses 24\,kbps bandwidth.

\begin{table*}[h]
\centering
\caption{\textbf{Comparison with baselines} on the music domain. Left: audio quality metrics (higher is better). Right: watermark detection measured by $\mathrm{TPR}@1\%\mathrm{FPR}$ (lower is better); \ding{55} indicates $\mathrm{TPR} < 0.01$.}

\label{tab:main_results_music_compact}
\setlength{\tabcolsep}{4pt}
\renewcommand{\arraystretch}{0.95}
\resizebox{0.8\textwidth}{!}{
\begin{tabular}{l l | c c | c c c c c}
\toprule
\multirow{2}{*}{\textbf{Type}} & \multirow{2}{*}{\textbf{Attack}}
& \multicolumn{2}{c|}{\textbf{Audio Quality}}
& \multicolumn{5}{c}{\textbf{Watermark Detection}($\mathrm{TPR}@1\%\mathrm{FPR}\downarrow$)}\\
\cmidrule(lr){3-4}\cmidrule(lr){5-9}
& & ViSQOL$\uparrow$ & MUSHRA$\uparrow$
& AudioSeal & WavMark & TimbreWM & Perth & SilentCipher \\
\midrule

\multirow{5}{*}{\textbf{Signal-level}}
& Pitch shift        & 4.061 & 78.33 & \ding{55} & \ding{55} & \ding{55} & 0.10 & \ding{55} \\
& Time stretch       & 4.109 & 80.16 & \ding{55} & 0.83 & 0.87 & 1.00 & \ding{55} \\
& Low-pass filter    & 3.203 & 87.28 & 1.00 & 1.00 & 1.00 & 1.00 & 0.51 \\
& High-pass filter   & 3.673 & 78.61 & 1.00 & 1.00 & 1.00 & 0.99 & 0.58 \\
& Additive noise     & 2.778 & 25.40 & 0.43 & \ding{55} & 0.12 & 0.43 & \ding{55} \\
\cmidrule(lr){1-9}

\multirow{2}{*}{\textbf{Codec}}
& MP3                & 4.505 & 94.39 & 1.00 & 0.84 & 0.94 & 1.00 & 0.42 \\
& EnCodec            & 4.321 & 88.66 & 0.99 & \ding{55} & \ding{55} & 0.89 & \ding{55} \\
\cmidrule(lr){1-9}

\multirow{3}{*}{\textbf{Adaptive}}
& Square Attack      & 2.960 & 31.24 & 0.18 & \ding{55} & 0.36 & 0.21 & \ding{55} \\

& \cellcolor{gray!15}\textbf{\textsc{DiffErase-latent}}
& \cellcolor{gray!15}4.163
& \cellcolor{gray!15}91.20
& \cellcolor{gray!15}\ding{55}
& \cellcolor{gray!15}\ding{55}
& \cellcolor{gray!15}0.01
& \cellcolor{gray!15}0.35
& \cellcolor{gray!15}\ding{55} \\

& \cellcolor{gray!15}\textbf{\textsc{DiffErase-mel}}
& \cellcolor{gray!15}3.938
& \cellcolor{gray!15}86.31
& \cellcolor{gray!15}\ding{55}
& \cellcolor{gray!15}\ding{55}
& \cellcolor{gray!15}0.01
& \cellcolor{gray!15}0.46
& \cellcolor{gray!15}\ding{55} \\

\bottomrule
\end{tabular}
}
\end{table*}
\vspace{-4mm}

\begin{table*}[h]
\centering
\caption{\textbf{Comparison with baselines} on the environment sound domain. Left: audio quality metrics (higher is better). Right: watermark detection measured by $\mathrm{TPR}@1\%\mathrm{FPR}$ (lower is better); \ding{55} indicates $\mathrm{TPR} < 0.01$.}

\label{tab:main_results_environment_compact}
\setlength{\tabcolsep}{4pt}
\renewcommand{\arraystretch}{0.95}
\resizebox{0.8\textwidth}{!}{
\begin{tabular}{l l | c c | c c c c c}
\toprule
\multirow{2}{*}{\textbf{Type}} & \multirow{2}{*}{\textbf{Attack}}
& \multicolumn{2}{c|}{\textbf{Audio Quality}}
& \multicolumn{5}{c}{\textbf{Watermark Detection}($\mathrm{TPR}@1\%\mathrm{FPR}\downarrow$)}\\
\cmidrule(lr){3-4}\cmidrule(lr){5-9}
& & ViSQOL$\uparrow$ & MUSHRA$\uparrow$
& AudioSeal & WavMark & TimbreWM & Perth & SilentCipher \\
\midrule

\multirow{5}{*}{\textbf{Signal-level}}
& Pitch shift        & 4.061 & 76.20 & \ding{55} & \ding{55} & \ding{55} & 0.03 & \ding{55} \\
& Time stretch       & 4.347 & 87.06 & \ding{55} & 0.85 & 0.97 & 1.00 & \ding{55} \\
& Low-pass filter    & 3.102 & 92.53 & 1.00 & 0.94 & 0.13 & 1.00 & 0.32 \\
& High-pass filter   & 3.706 & 84.80 & 1.00 & 1.00 & 0.97 & 1.00 & 0.59 \\
& Additive noise     & 1.443 & 31.43 & 0.34 & \ding{55} & 0.11 & 0.39 & \ding{55} \\
\cmidrule(lr){1-9}

\multirow{2}{*}{\textbf{Codec}}
& MP3                & 3.979 & 97.21 & 1.00 & 0.65 & 0.82 & 0.99 & 0.22 \\
& EnCodec            & 4.334 & 94.81 & 0.96 & \ding{55} & \ding{55} & 0.90 & \ding{55} \\
\cmidrule(lr){1-9}

\multirow{3}{*}{\textbf{Adaptive}}
& Square Attack      & 2.365 & 58.20 & 0.11 & \ding{55} & 0.06 & 0.38 & \ding{55} \\

& \cellcolor{gray!15}\textbf{\textsc{DiffErase-latent}}
& \cellcolor{gray!15}3.308
& \cellcolor{gray!15}87.07
& \cellcolor{gray!15}\ding{55}
& \cellcolor{gray!15}\ding{55}
& \cellcolor{gray!15}\ding{55}
& \cellcolor{gray!15}0.23
& \cellcolor{gray!15}\ding{55} \\

& \cellcolor{gray!15}\textbf{\textsc{DiffErase-mel}}
& \cellcolor{gray!15}3.952
& \cellcolor{gray!15}86.48
& \cellcolor{gray!15}\ding{55}
& \cellcolor{gray!15}\ding{55}
& \cellcolor{gray!15}\ding{55}
& \cellcolor{gray!15}0.19
& \cellcolor{gray!15}\ding{55} \\
\bottomrule
\end{tabular}
}
\end{table*}

\paragraph{Adaptive attack.}
We implement Square attack~\citep{andriushchenko2020square} in the spectrogram domain following \citet{liu2024audiomarkbench}. For speech, we use a query budget of 10,000 and perturbation bound $\epsilon=0.02$. For music and environmental sounds, we found $\epsilon=0.02$ insufficient for successful attacks and increased it to $\epsilon=0.2$.

\subsection{Dataset details}\label{apx:dataset_detail}
We evaluate DiffErase across three audio domains to assess generalization. For speech, we use LibriSpeech~\citep{panayotov2015librispeech}, an English corpus derived from audiobooks containing 1,000 hours of speech. We train on the \textit{train-clean-100} subset (100 hours) and evaluate on 100 randomly sampled clips from \textit{test-clean}. For environmental sounds, we use Clotho~\citep{drossos2020clotho}, an audio captioning dataset. We discard the captions and reserve 100 samples for evaluation, using the reminder for training. For music, we use \textit{FMA-small}~\citep{defferrard2016fma} from the Free Music Archive, which contains songs across various genres. We use all samples for training except 100 samples for evaluation.

\subsection{Subjective listening test}
\label{apx:subjective_test}
We conduct a MUSHRA listening test following ITU-R BS.1534 to evaluate perceptual quality. We randomly sample 5 audio clips per domain and include all attack samples in each trial. 18 participants completed the study, filtered unreliable evaluations using low-quality anchor samples, and 16 valid participants were retained for analysis. Participants rate the quality of processed audio samples on a scale of 0 to 100, with watermarked audio provided as the reference. Evaluations are conducted using the open-source webMUSHRA platform~\citep{schoeffler2018webmushra} locally.

\section{Additional experimental results}\label{apx:training_detail}
\label{apx:additional_results}
\subsection{Comparison with baselines on music and environmental sounds}\label{apx:comparison_on_music_env}
Table~\ref{tab:main_results_music_compact} and Table~\ref{tab:main_results_environment_compact} present comparison results on the music and environmental sound domains, respectively. The overall performance is consistent with the speech domain. 

Signal-level attacks show limited effectiveness. Pitch shifting removes most watermarks but degrades quality (MUSHRA of 78.33 on music and 76.20 on environmental). Additive noise achieves partial removal but introduces severe distortions (MUSHRA of 25.40 on music and 31.43 on environment). Time stretching and frequency filtering remain largely ineffective.

Codec-based attacks preserve high quality, but fail to remove watermarks. MP3 compression maintains MUSHRA above 94 but fails to remove most watermarks. EnCodec successfully removes WavMark, TimbreWM, and SilentCipher while preserving audio quality, but fails against AudioSeal and Perth.

Even though Square attack shows relatively good removal performance on speech, it fails on music and environmental sounds despite increasing the perturbation bound to $\epsilon=0.2$, which causes substantial quality degradation (MUSHRA of 31.24 on music and 58.20 on environment).

DiffErase achieves strong watermark removal across both domains while preserving competitive quality. On music, \textsc{DiffErase-latent} achieves ViSQOL of 4.16 and MUSHRA of 91.20, removing all watermarks except Perth ($\mathrm{TPR}=0.35$). On environmental sounds, \textsc{DiffErase-mel} achieves ViSQOL of 3.95 and MUSHRA of 86.48. Perth remains partially detectable owing to its stronger embedded perturbations observed in Figure~\ref{fig:l2_distance}.

\begin{figure*}[t]
    \centering
    \includegraphics[width=0.32\textwidth]{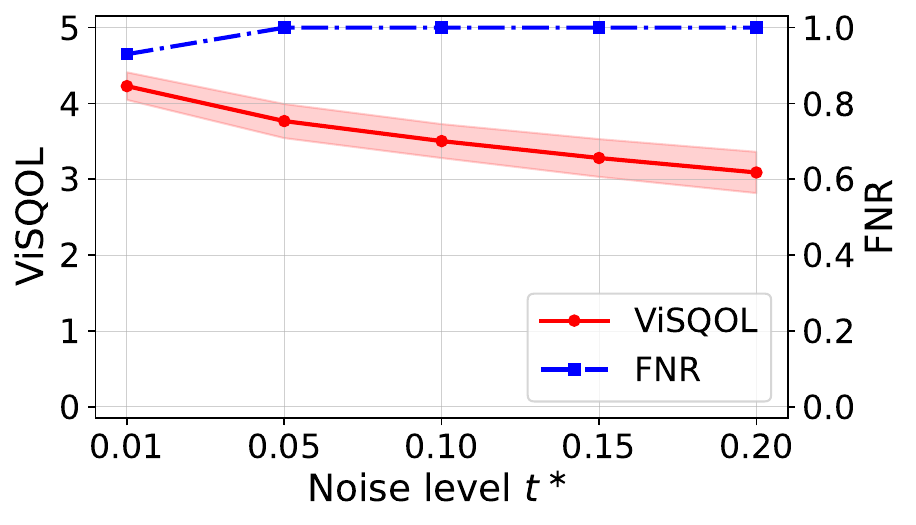}\hfill
    \includegraphics[width=0.32\textwidth]{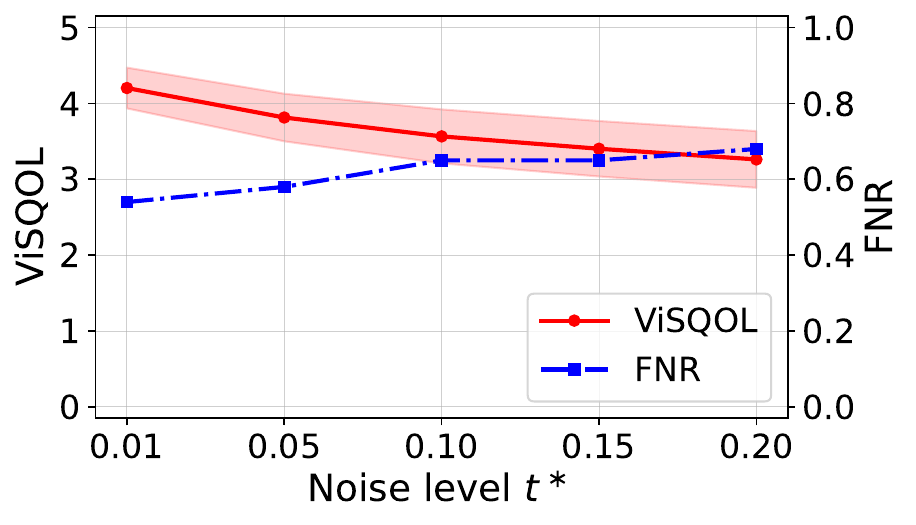}\hfill
    \includegraphics[width=0.32\textwidth]{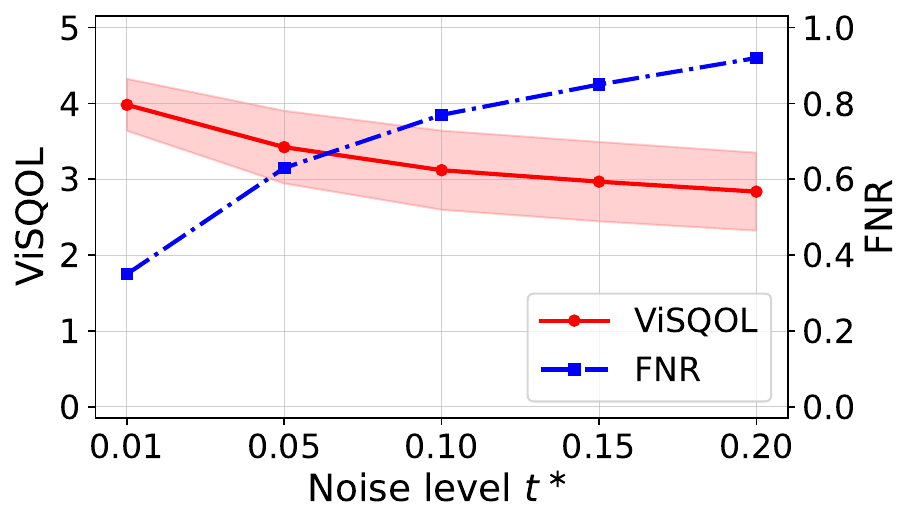}
    \caption{\textbf{Effect of noise level $t^*$ for \textsc{DiffErase-latent}.} Trade-off between audio quality (ViSQOL, left axis) and watermark removal (FNR $= 1 - \mathrm{TPR}$, right axis), evaluated on Perth. \textbf{Left:} Speech. \textbf{Middle:} Music. \textbf{Right:} Environment.}
    \label{fig:noise_level_tradeoff_latent}
\end{figure*}

\subsection{Noise level trade-off for \textsc{DiffErase-latent}}~\label{apx:noise_level_tradeoff_latent}
Figure~\ref{fig:noise_level_tradeoff_latent} shows the quality-removal trade-off for \textsc{DiffErase-latent}. The trends are similar: increasing $t^\ast$ improves watermark removal (higher $\mathrm{FNR}$) while reducing audio quality (lower ViSQOL). On speech, Perth becomes undetectable at $t^\ast \ge 0.05$. On music and environmental sounds, complete removal requires higher noise levels. Compared to \textsc{DiffErase-mel}, \textsc{DiffErase-latent} achieves comparable removal at similar noise levels but with slightly lower ViSQOL, likely due to the additional compression from the VAE encoder.

\begin{table}[t]
\centering

\begin{minipage}[t]{0.49\columnwidth}
\centering
\caption{\textbf{Ablation study} on music.}
\label{tab:differase_ablation_music}
\footnotesize
\setlength{\tabcolsep}{4pt}
\renewcommand{\arraystretch}{0.95}
\resizebox{\linewidth}{!}{
\begin{tabular}{l c c c}
\toprule
\textbf{Method} & \textbf{TimbreWM} & \textbf{Perth} & \textbf{WavMark} \\
\midrule
\multicolumn{4}{l}{\textbf{(i) mel-to-waveform}} \\
\addlinespace[1pt]
Griffin-Lim algorithm (GLA) only                       & 1.00 & 1.00 & 1.00 \\
\textsc{DiffErase-mel} + GLA & 0.00 & 0.46 & 0.00  \\
\textsc{DiffErase-latent} + GLA & 0.00 & 0.00 & 0.00  \\
\cmidrule(lr){1-4}
\multicolumn{4}{l}{\textbf{(ii) diffusion sampler}} \\
\addlinespace[1pt]
\textsc{DiffErase-mel} (DDPM)  & 0.01 & 0.46 & 0.00 \\
\textsc{DiffErase-mel} (DDIM)  & 0.02 & 0.75 & 0.00 \\
\textsc{DiffErase-latent} (DDPM)  & 0.01 & 0.35 & 0.00 \\
\textsc{DiffErase-latent} (DDIM)  & 0.02 & 0.44 & 0.00 \\
\bottomrule
\end{tabular}
}
\end{minipage}\hfill
\begin{minipage}[t]{0.49\columnwidth}
\centering
\caption{\textbf{Ablation study} on environment sound.}
\label{tab:differase_ablation_env}
\footnotesize
\setlength{\tabcolsep}{4pt}
\renewcommand{\arraystretch}{0.95}
\resizebox{\linewidth}{!}{
\begin{tabular}{l c c c}
\toprule
\textbf{Method} & \textbf{TimbreWM} & \textbf{Perth} & \textbf{WavMark} \\
\midrule
\multicolumn{4}{l}{\textbf{(i) mel-to-waveform}} \\
\addlinespace[1pt]
Griffin-Lim algorithm (GLA) only                       & 0.89 & 1.00 & 1.00 \\
\textsc{DiffErase-mel} + GLA & 0.00 & 0.58 & 0.00  \\
\textsc{DiffErase-latent} + GLA & 0.00 & 0.00 & 0.00  \\
\cmidrule(lr){1-4}
\multicolumn{4}{l}{\textbf{(ii) diffusion sampler}} \\
\addlinespace[1pt]
\textsc{DiffErase-mel} (DDPM)  & 0.00 & 0.19 & 0.00 \\
\textsc{DiffErase-mel} (DDIM)  & 0.00 & 0.34 & 0.00 \\
\textsc{DiffErase-latent} (DDPM)  & 0.00 & 0.23 & 0.00 \\
\textsc{DiffErase-latent} (DDIM)  & 0.00 & 0.23 & 0.00 \\
\bottomrule
\end{tabular}
}
\end{minipage}
\end{table}

\begin{table}[t]
\centering
\caption{\textbf{Ablation on audio representation.} We compare DiffErase with diffusion on different representations.}
\label{tab:repr_ablation}
\setlength{\tabcolsep}{4pt}
\renewcommand{\arraystretch}{0.95}
\resizebox{0.8\columnwidth}{!}{
\begin{tabular}{l | c c c | c c c c c}
\toprule
\multirow{2}{*}{\textbf{Representation}}
& \multicolumn{3}{c|}{\textbf{Audio Quality}}
& \multicolumn{5}{c}{\textbf{Watermark Detection}($\mathrm{TPR}@1\%\mathrm{FPR}\downarrow$)}\\
\cmidrule(lr){2-4}\cmidrule(lr){5-9}
& SQUIM-MOS$\uparrow$ & ViSQOL$\uparrow$ & MUSHRA$\uparrow$
& AudioSeal & WavMark & TimbreWM & Perth & SilentCipher \\
\midrule
Waveform           & 2.541 & 1.996 & 71.21 & \ding{55} & \ding{55} & \ding{55} & \ding{55} & \ding{55} \\
Linear spectrogram & 3.825 & 3.857 & 79.66 & 1.00 & 1.00 & 1.00 & 0.97 & 0.35 \\
\rowcolor{gray!15}
Mel-spectrogram    & 4.423 & 3.961 & 93.81 & \ding{55} & \ding{55} & \ding{55} & \ding{55} & \ding{55} \\
\rowcolor{gray!15}
Mel-latent         & 4.214 & 3.477 & 87.73 & \ding{55} & \ding{55} & \ding{55} & \ding{55} & \ding{55} \\
\bottomrule
\end{tabular}
}
\end{table}

\subsection{Ablation study on music and environmental sounds}\label{apx:more_ablation}

Tables~\ref{tab:differase_ablation_music} and~\ref{tab:differase_ablation_env} extend the ablation study to the music and environmental sound domains. Consistent with the speech results, Griffin--Lim reconstruction alone does not remove watermarks ($\mathrm{TPR} = 1.00$ for most systems). Adding diffusion substantially improves removal. DDPM consistently outperforms DDIM across both domains, particularly for Perth. On music, DDPM reduces Perth to $\mathrm{TPR} = 0.46$ while DDIM only achieves 0.75. This confirms that the finer-grained denoising trajectory of DDPM better suppresses watermark residues.

\subsection{Ablation on attack representation}\label{apx:ablation_on_representation}
We compare DiffErase operating on mel-spectrograms against alternative representations: (i) waveform using DiffWave~\citep{kong2020diffwave}, and (ii) linear spectrogram with phase preservation following WavePurifier~\citep{guo2024wavepurifier}. 

As shown in Table~\ref{tab:repr_ablation}, waveform diffusion successfully removes watermark but produces low-quality audio (SQUIM-MOS of 2.541, MUSHRA of 71.21). The high dimensionality of raw waveforms leads to over-smoothed content and temporal drift, failing to preserve perceptual details. This is evident in Figure~\ref{fig:ablation_spec}(c), where harmonic structures are blurred, and in Figure~\ref{fig:ablation_wav}(b), where the reconstructed waveform significantly deviates from the original envelope. Linear spectrogram diffusion improves quality but remains inferior to the mel-spectrogram. Its reconstruction relies on reusing the original phase, but the regenerated magnitude may be inconsistent with the preserved phase, producing audible artifacts, as shown in Figure~\ref{fig:ablation_spec}(d) and Figure~\ref{fig:ablation_wav}(c). In contrast, mel-spectrogram diffusion achieves the best trade-off: complete watermark removal while maintaining high audio quality. As shown in Figure~\ref{fig:ablation_spec}(e), DiffErase preserves the main spectral structure while removing watermark perturbations, enabling stable diffusion and high-quality reconstruction via modern neural vocoders. This is further confirmed in Figure~\ref{fig:ablation_wav}(d), where the reconstructed waveform closely follows the original temporal envelope.

\begin{figure}[h]
\centering
\begin{subfigure}[t]{0.18\textwidth}
    \centering
    \includegraphics[width=0.9\linewidth]{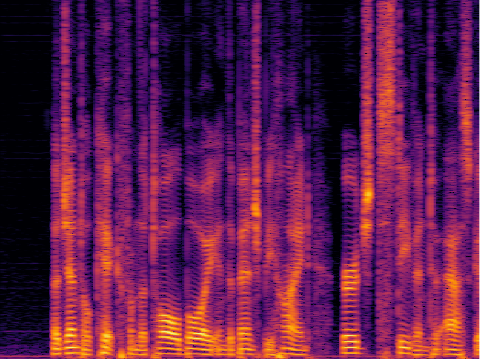}
    \caption*{(a) Original}
\end{subfigure}
\hfill
\begin{subfigure}[t]{0.18\textwidth}
    \centering
    \includegraphics[width=0.9\linewidth]{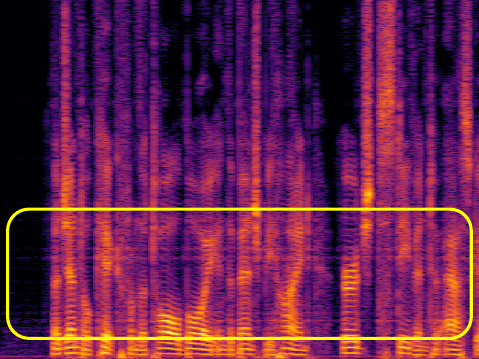}
    \caption*{(b) Watermarked}
\end{subfigure}
\hfill
\begin{subfigure}[t]{0.18\textwidth}
    \centering
    \includegraphics[width=0.9\linewidth]{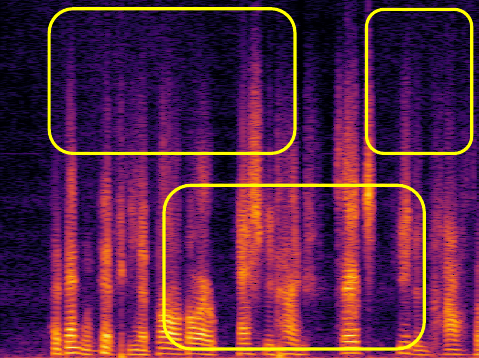}
    \caption*{(c) DiffWave}
\end{subfigure}
\hfill
\begin{subfigure}[t]{0.18\textwidth}
    \centering
    \includegraphics[width=0.9\linewidth]{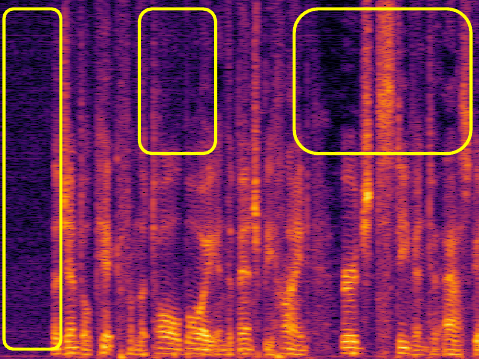}
    \caption*{(d) WavePurifier}
\end{subfigure}
\hfill
\begin{subfigure}[t]{0.18\textwidth}
    \centering
    \includegraphics[width=0.9\linewidth]{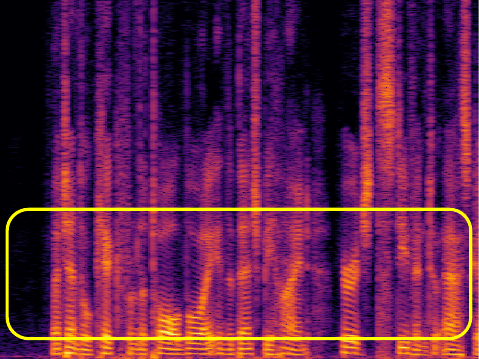}
    \caption*{(e) DiffErase}
\end{subfigure}

\caption{\textbf{Mel-spectrogram visualization} across different attack representations. DiffWave (c) over-smooths harmonic details. WavePurifier (d) introduces audible artifacts. DiffErase (e) preserves the spectral structure while removing watermark perturbations.}
\label{fig:ablation_spec}
\end{figure}

\begin{figure}[h]
\centering
\begin{subfigure}[t]{0.48\textwidth}
    \centering
    \includegraphics[width=1\linewidth]{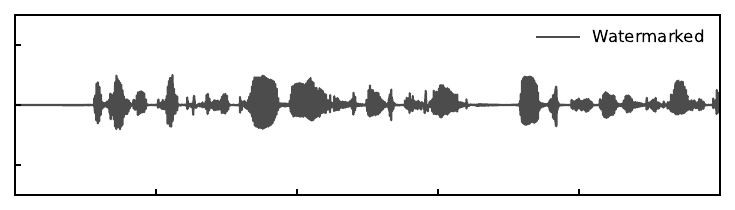}
    \caption*{(a) Watermarked}
\end{subfigure}
\hfill
\begin{subfigure}[t]{0.48\textwidth}
    \centering
    \includegraphics[width=1\linewidth]{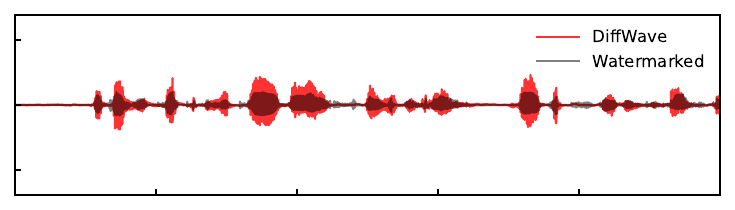}
    \caption*{(b) DiffWave}
\end{subfigure}

\vspace{2mm}

\begin{subfigure}[t]{0.48\textwidth}
    \centering
    \includegraphics[width=1\linewidth]{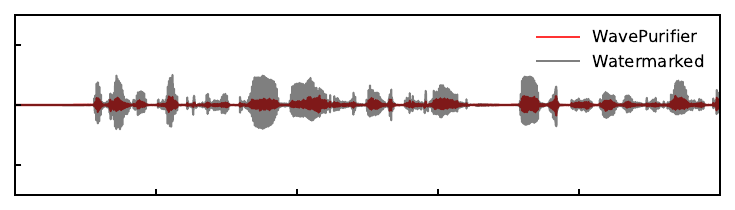}
    \caption*{(c) WavePurifier}
\end{subfigure}
\hfill
\begin{subfigure}[t]{0.48\textwidth}
    \centering
    \includegraphics[width=1\linewidth]{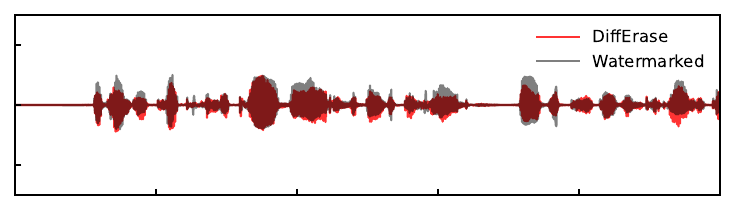}
    \caption*{(d) DiffErase}
\end{subfigure}

\caption{\textbf{Waveform visualization} across different attack representations. Red waveforms show reconstructed audio, and gray shows input watermarked audio. DiffErase (d) closely matches the original temporal envelope.}
\label{fig:ablation_wav}
\end{figure}

\end{document}